\documentclass[11pt]{article}
\usepackage[margin=1in]{geometry}
\usepackage{complexity}
\usepackage{mathtools,amsmath,amssymb}
\usepackage{xspace}
\RequirePackage[colorlinks=true]{hyperref}
\hypersetup{
  linkcolor=[rgb]{0,0,0.5},
  citecolor=[rgb]{0, 0.5, 0},
  urlcolor=[rgb]{0.5, 0, 0}
}
\usepackage{dsfont}

\usepackage{amsthm}
\usepackage{thmtools,thm-restate}

\numberwithin{equation}{section}
\declaretheoremstyle[bodyfont=\it,qed=\qedsymbol]{noproofstyle}

\declaretheorem[numberlike=equation]{observation}

\declaretheorem[name=Observation,numbered=no]{observation*}

\declaretheorem[numberlike=equation]{theorem}

\declaretheorem[name=Theorem,numbered=no]{theorem*}

\declaretheorem[numberlike=equation]{lemma}
\declaretheorem[name=Lemma,numbered=no]{lemma*}

\declaretheorem[numberlike=equation]{corollary}
\declaretheorem[name=Corollary,numbered=no]{corollary*}

\declaretheorem[name=Proposition,numbered=no]{proposition*}

\declaretheorem[numberlike=equation]{claim}
\declaretheorem[name=Claim,numbered=no]{claim*}

\declaretheorem[numberlike=equation]{conjecture}
\declaretheorem[name=Conjecture,numbered=no]{conjecture*}

\declaretheorem[name=Question,numbered=no]{question*}

\declaretheoremstyle[bodyfont=\it,qed=$\lozenge$]{defstyle} 

\declaretheorem[numberlike=equation,style=defstyle]{definition}
\declaretheorem[unnumbered,name=Definition,style=defstyle]{definition*}

\declaretheorem[unnumbered,name=Example,style=defstyle]{example*}

\declaretheorem[unnumbered,name=Notation=defstyle]{notation*}

\declaretheorem[unnumbered,name=Construction,style=defstyle]{construction*}

\declaretheorem[numberlike=equation,style=defstyle]{remark}
\declaretheorem[unnumbered,name=Remark,style=defstyle]{remark*}

\usepackage{nth}
\usepackage{intcalc}
\usepackage{etoolbox}
\usepackage{xstring}
\hypersetup{
}

\usepackage{ifpdf}
\ifpdf
\else
\usepackage[quadpoints=false]{hypdvips}
\fi

\newcommand{\ECCCurl}[2]{http://eccc.weizmann.ac.il/report/\ifnumcomp{#1}{>}{93}{19}{20}#1/#2/}

\newcommand{\shortECCC}[2]{\texttt{\href{http://eccc.weizmann.ac.il/report/\ifnumcomp{#1}{>}{93}{19}{20}#1/#2/}{eccc:TR#1-#2}}}

\newcommand{\parseECCC}[1]{% Takes a string of the form TRxx/xxx or
\StrSubstitute{#1}{TR}{}[\tmpstring]%
\IfSubStr{\tmpstring}{/}{ %assuming string is of the form TRxx/xxx
\StrBefore{\tmpstring}{/}[\ecccyear]%
\StrBehind{\tmpstring}{/}[\ecccreport]%
}{% assuming string is of the form TRxx-xxx
\StrBefore{\tmpstring}{-}[\ecccyear]%
\StrBehind{\tmpstring}{-}[\ecccreport]%
}%
\shortECCC{\ecccyear}{\ecccreport}}

	\renewcommand{\vec}[1]{{\mathbf{#1}}}

	\makeatletter
	\newcommand{\va}{{\vec{a}}\@ifnextchar{^}{\!\:}{}}
	\newcommand{\vb}{{\vec{b}}\@ifnextchar{^}{\!\:}{}}
	\newcommand{\vc}{{\vec{c}}\@ifnextchar{^}{\!\:}{}}
	\newcommand{\vd}{{\vec{d}}\@ifnextchar{^}{\!\:}{}}
	\newcommand{\ve}{{\vec{e}}\@ifnextchar{^}{\!\:}{}}
	\newcommand{\vy}{{\vec{y}}\@ifnextchar{^}{\!\:}{}}
	\newcommand{\vs}{{\vec{s}}\@ifnextchar{^}{\!\:}{}}
	\newcommand{\vt}{{\vec{t}}\@ifnextchar{^}{\!\:}{}}
	\newcommand{\vx}{{\vec{x}}\@ifnextchar{^}{}{}}		%\vec{x} seems fine already
	\newcommand{\vz}{{\vec{z}}\@ifnextchar{^}{\!\:}{}}

	\newcommand{\vY}{{\vec{Y}}\@ifnextchar{^}{\!\:}{}}
	\newcommand{\vX}{{\vec{X}}\@ifnextchar{^}{}{}}		%\vec{x} seems fine already
	\newcommand{\vZ}{{\vec{Z}}\@ifnextchar{^}{\!\:}{}}
	\newcommand{\vG}{{\vec{G}}\@ifnextchar{^}{\!\:}{}}
	
	\makeatother

\newcommand{\F}{\mathbb{F}}

\newcommand{\Q}{\mathbb{Q}}
\newcommand{\N}{\mathbb{N}}

\newcommand{\h}{\mathcal{H}}

\newcommand{\set}[1]{\left\{#1\right\}}

\newcommand{\abs}[1]{\left|#1\right|}

\DeclareMathOperator{\size}{size}

\newcommand{\sps}{\sum\prod\sum}

\def\epsilon{\varepsilon}

\date{}
\title{Some Closure Results for Polynomial Factorization and Applications}
\author{
Chi-Ning Chou\thanks{School of Engineering and Applied Sciences, Harvard University, Cambridge, Massachusetts, USA. Email: \texttt{chiningchou@g.harvard.edu}.}
\and
Mrinal Kumar\thanks{Center for Mathematical Sciences and Applications, Harvard University, Cambridge, Massachusetts, USA. Email: \texttt{mrinalkumar08@gmail.com}.}
\and
Noam Solomon\thanks{Center for Mathematical Sciences and Applications, Harvard University, Cambridge, Massachusetts, USA. Email: \texttt{noam.solom@gmail.com}.}
}
\begin{document}
\maketitle

\begin{abstract}
In a sequence of  fundamental results in the 80's, Kaltofen~\cite{k85, Kal86, Kal87, k89} showed that factors of multivariate polynomials with small arithmetic circuits have small arithmetic circuits. In other words, the complexity class $\VP$ is closed under taking factors. A natural question in this context is to understand if other natural classes of multivariate polynomials, for instance, arithmetic formulas, algebraic branching programs, bounded depth arithmetic circuits or the class $\VNP$, are closed under taking factors. 

In this paper, we show that all factors of degree at most $\log^a n$ of polynomials with $\poly(n)$ size depth $k$ circuits have $\poly(n)$ size circuits of depth at most $O(k + a)$. This partially answers a  question of Shpilka-Yehudayoff (Q. 19 in~\cite{sy}) and has applications to hardness-randomness tradeoffs for bounded depth arithmetic circuits.  

More precisely, this shows that  a superpolynomial lower bound for bounded depth arithmetic circuits, for a family of explicit polynomials of degree $\poly(\log n)$ implies deterministic sub-exponential time algorithms for polynomial identity testing (PIT) for bounded depth arithmetic circuits. This is incomparable to a beautiful result of Dvir et al.~\cite{DSY09}, where they showed that super-polynomial lower bounds for  constant depth arithmetic circuits for any explicit family of polynomials (of potentially high degree) implies sub-exponential time deterministic PIT for bounded depth circuits of \emph{bounded individual degree}. Thus, we remove the ``bounded individual degree" condition in~\cite{DSY09} at the cost of strengthening the hardness assumption to hold for polynomials of \emph{low} degree. 

As direct applications of our techniques, we also obtain simple proofs of the following results.
\begin{itemize}
\item The complexity class $\VNP$ is closed under taking factors. This confirms a conjecture of B{\"u}rgisser (Conj. 2.1 in~\cite{bur00}), and improves upon a recent result of Dutta, Saxena and Sinhababu~\cite{DSS17} who showed a quasipolynomial upper bound on the number of auxiliary variables and the complexity of the verifier circuit of factors of polynomials in $\VNP$. 

\item A factor of degree at most $d$ of a polynomial $P$ which can be computed by an arithmetic formula (resp. algebraic branching program) of size $s$ has a formula (resp. algebraic branching program) of size at most $\poly(s, d^{\log d}, \deg(P))$. This result was first shown by Dutta et al.~\cite{DSS17}, and we obtain a slightly different proof  as an easy consequence of our techniques. 
\end{itemize}
Our proofs rely on a combination of the \emph{lifting} based ideas developed in polynomial factoring literature and the depth reduction results for arithmetic circuits, and hold over fields of characteristic zero or sufficiently large.

%In this paper, we study the question of hardness-randomness tradeoffs for bounded depth arithmetic circuits. 

%The key technical ingredient of our proof is the following property of roots of polynomials computable by a bounded depth arithmetic circuit : if $f(x_1, x_2, \ldots, x_n)$ and \linebreak $P(x_1, x_2, \ldots, x_n, y)$ are polynomials of degree $d$ and $r$ respectively, such that $P$ can be computed by a circuit of size $s$ and depth $\Delta$ and $P(x_1, x_2, \ldots, x_n, f) \equiv 0$, then, $f$ can be computed by a circuit of size $\poly(n, s, r, d^{O(\sqrt{d})})$ and depth $\Delta + 3$. In comparison, Dvir et al.~\cite{DSY09} showed that $f$ can be computed by a circuit of depth $\Delta + 3$ and  size $\poly(n, s, r, d^{t})$, where $t$ is the degree of $P$ in $y$. Thus, the size upper bound in~\cite{DSY09} is non-trivial when $t$ is small but $d$ could be large, whereas our size upper bound is non-trivial when $d$ is small, but $t$ could be large. 

\end{abstract}

\thispagestyle{empty}
\newpage
\pagenumbering{arabic}
\section{Introduction}
A fundamental question in computational algebra is the question of polynomial factorization : Given a polynomial $P$, can we efficiently compute the factors of $P$? In this paper, we will be interested in the following closely related question : Given a \emph{structured} polynomial $P$, what can we say about the structure of factors of $P$? 
 
In a sequence of seminal results, Kaltofen~\cite{k85, Kal86, Kal87, k89} showed that  if a polynomial $P$ of degree $d$ in $n$ variables has an arithmetic circuit of size $s$, then each of its factors has an arithmetic circuit of size $\poly(s, n, d)$. Moreover, he also showed that given the circuit for $P$, the circuits for its factors can be computed in time $\poly(s, n, d)$ by a randomized algorithm. 

Another way of stating this result is that the complexity class $\VP$, which we now define, is \emph{uniformly closed under taking factors}.
\begin{definition}[VP]\label{def:VP}
	A family of polynomials $\{f_n\}$ over a field $\F$ is said to be in the class $\VP_{\F}$ if there exist polynomially bounded functions $d, k, v : \N\rightarrow \N$ and a circuit family $\{g_n\}$ such that $\deg(f_n)\leq d(n)$, $\size(g_n)\leq s(n)$, and $f_n$ is computed by $g_n$ for every sufficiently large $n \in \N$.
\end{definition}
We remark that factorization is a fundamental algebraic notion, and so closure under factorization indicates that a complexity class is algebraically nice in some sense. Thus, it is a natural question to ask if any of the other naturally and frequently occurring classes of polynomials like $\class{VF}$ (polynomials with small formulas), $\class{VBP}$ (polynomials with small algebraic branching programs), constant depth arithmetic circuits, or the class $\VNP$ (the algebraic analog of $\NP$ or $\#\P$) are closed under taking factors. 

In recent years, we have had some progress on the question of closure under factorization 
for bounded depth arithmetic circuits (see \cite{DSY09, O16}) or the classes $\class{VF}, \class{VBP}$ and $\VNP$ (see~\cite{DSS17}). We will discuss these results in a later part of this section.

In addition to being basic questions in algebraic complexity, some of these closure results also have applications to extending the hardness vs randomness framework of Kabanets and Impagliazzo~\cite{ki03} to formulas, branching programs or bounded depth arithmetic circuits. Indeed, Kaltofen's closure result for arithmetic circuits is crucial ingredient in the proof of Kabanets and Impagliazzo~\cite{ki03}. 
 
\subsection{Hardness and Randomness}
Two of the most basic questions in algebraic complexity theory are the question of proving super-polynomial lower bounds on the size of arithmetic circuits computing some explicit polynomial family, and that of designing efficient deterministic algorithms for Polynomial Identity Testing (PIT). 

The progress on these questions for general arithmetic circuits has been painfully slow. To date, there are no non-trivial algorithms for PIT for general arithmetic circuits, while the best known lower bound, due to Bauer and Strassen~\cite{BS83}, is a slightly superlinear lower bound $\Omega(n\log n)$, proved over three decades ago. In fact, even for the class of bounded depth arithmetic circuits, no non-trivial deterministic PIT algorithms are known, and the best lower bounds known are just slightly superlinear~\cite{Raz10a}. 

In a very influential work, Kabanets and Impagliazzo~\cite{ki03} showed that the questions of derandomizing PIT and that of proving lower bounds for arithmetic circuits are equivalent in some sense. Their result adapts the Hardness vs Randomness framework of Nisan and Wigderson~\cite{nw94} to the algebraic setting. In their proof, Kabanets and Impagliazzo combine the use of Nisan-Wigderson generator with Kaltofen's result that all factors of a low degree (degree $\poly(n)$) polynomial with $\poly(n)$ sized circuit are computable by size $\poly(n)$ circuits~\cite{k89}. They showed that given an explicit family of \emph{hard} polynomials, one can obtain a \emph{non-trivial}\footnote{Here, non-trivial means subexponential time, or quasipolynomial time, based on the hardness assumption.} deterministic algorithm for PIT. 

The extremely slow progress on the lower bound and PIT questions for general circuits has led to a lot of attention on understanding these questions for more structured sub-classes of arithmetic circuits. Arithmetic formula~\cite{k85}, algebraic branching programs~\cite{k17}, multilinear circuits~\cite{Raz06, raz-yehudayoff, RSY08}, and constant depth arithmetic circuits~\cite{nw1997, Raz10a, GKKS14, FLMS13, KS14} are some examples of such circuit classes. An intriguing question is to ask if the equivalence of PIT and lower bounds also carries over to these more structured circuit classes. For example, does super-polynomial lower bounds for arithmetic formulas imply non-trivial deterministic algorithms for PIT for arithmetic formulas, and vice-versa? 

The answers to these questions do not follow directly from the results in~\cite{ki03}; and extending the approach of Kabanets and Impagliazzo to answer these questions seems to be intimately related to the questions about closure of arithmetic formulas and bounded depth circuits under polynomial factorization. 

We now describe our results, and discuss how they relate to prior work.
\section{Results and Prior Work}\label{sec:results and prior work}

\subsection{Factors of Polynomials with Bounded Depth Circuits}
For our first set of results, we study the bounded depth circuit complexity of factors of polynomials which have small bounded depth circuits. We prove the following theorem. 
\begin{theorem}\label{thm:factor ckt ub-intro}
Let $\F$ be a field of characteristic zero. Let $P \in \F[\vx]$ be a polynomial of degree at most $r$ in $n$ variables that can be computed by an arithmetic circuit of size $s$ of depth at most $\Delta$. Let $f\in \F[\vx]$ be an irreducible  polynomial of degree at most $d$ such that $f$ divides $P$. Then, $f$ can be computed by a circuit of depth at most $\Delta + O(1)$ and size at most $O(\poly(s, r, n)\cdot d^{O(\sqrt{d})})$. 
\end{theorem}
Thus, low degree factors of polynomials with small low depth circuits have small low depth circuits. Our proof gives a smooth tradeoff in the depth of the circuit for the factor and its size. The tradeoff is governed by the depth reduction results for arithmetic circuits (see~\autoref{thm:depth-2k chasm}). 
We remark that the result is also true when the characteristic of the underlying field is sufficiently large. 
The result in the literature, which is most closely related to~\autoref{thm:factor ckt ub-intro}, is due to Oliveira~\cite{O16}. He studied the question of bounded depth circuit complexity of factors of polynomials with small bounded depth circuits, for polynomials of low individual degree. He showed that if a polynomial $P$ of individual degree $r$ is computable by a circuit of size $s$ and depth $\Delta$, then every factor of $P$ of degree at most $d$ can be computed by a circuit of size $\poly(s, r, d^r)$ and depth at most $\Delta + 5$. Thus, for polynomials with small individual degree, the results in~\cite{O16} are strictly better than ours, whereas for polynomials with unbounded individual degree, we get a better upper bound on the complexity of factors of total degree at most $\poly(\log n)$.

One of our main motivations for studying this question is the connection to hardness-randomness tradeoffs for bounded depth arithmetic circuits. In the next section, we describe the implications of our results in this context. 

\subsection{Hardness vs Randomness for Bounded Depth Circuits}
Dvir, Shpilka and Yehudayoff~\cite{DSY09} initiated the study of the question of the equivalence between PIT and lower bounds for bounded depth circuits. Dvir et al. observed that a part of the proof in~\cite{ki03} can be generalized to show that non-trivial PIT for bounded depth circuits implies lower bounds for such circuits. For the converse, the authors only showed a weaker statement; they proved that super-polynomial lower bounds for  depth $\Delta$ arithmetic circuits implies non-trivial PIT for depth $\Delta-5$ arithmetic circuits with \emph{bounded individual degree}. The bounded individual degree condition is a bit unsatisfying, and so, the following question is of interest : Does a super-polynomial lower bound for depth $\Delta$ arithmetic circuits imply non-trivial deterministic PIT for depth $\Delta'$ arithmetic circuits\footnote{Here, we think of $\Delta'$ as $\Delta - O(1)$.}? In particular, can we get rid of the ``bounded individual degree'' condition from the results in~\cite{DSY09}?

In this paper, we partially answer this question in the affirmative. Informally, we prove the following theorem. 
\begin{theorem}[Informal]\label{thm:main informal}
A super-polynomial lower bound for depth $\Delta$ arithmetic circuits for an explicit family of \emph{low degree} polynomials implies non-trivial deterministic PIT for depth $\Delta-5$ arithmetic circuits. 
\end{theorem}
Here, by low degree polynomials, we mean polynomials in $n$ variables and degree at most  $O(\log^2n/\log^2\log n)$. Thus, by strengthening the hardness hypothesis in~\cite{DSY09}, we remove the bounded individual degree restriction from the implication. We now state the result in~\autoref{thm:main informal} formally.
\begin{theorem}\label{thm:main thm}
Let $\Delta \geq 6$ be a positive integer, and let $\epsilon > 0$ be any real number. Let $\{f_m\}$ be a family of explicit polynomials such that $f_m$ is an $m$-variate multilinear polynomial of degree $d = O\left({\log^2 m}/{\log^2\log m}\right)$ which cannot be computed by an arithmetic circuit of depth $\Delta$ and size $\poly(m)$. Then, there is a deterministic algorithm, which, given as input a circuit $C \in \Q[\vx]$ of size $s$, depth $\Delta-5$ and degree $D$ on $n$ variables, runs in time $(snD)^{O(n^{2\epsilon})}$ and determines if the polynomial computed by $C$ is identically zero. 
\end{theorem}
Some remarks on the above theorem statement. 

\begin{remark}
The running time of the PIT algorithm gets better as the lower bound gets stronger. Also, the constraint on the degree of the hard polynomial family can be further relaxed a bit, at the cost of strengthening the hardness assumption, and increasing the running time of the resulting PIT algorithm\footnote{{If we assume sub-exponential lower bound, then we can get a quasi-polynomial time PIT. Note that this is the parameter region used in~\cite{DSY09}}}. We leave it to the interested reader to work out these details. We also note that the multilinearity assumption on the hard polynomial family is  without loss of generality.
\end{remark}

\iffalse
\begin{remark}
In general, explicit polynomial families do not have to be multilinear. But, if we have a hard
polynomial which is not multilinear, and has a polynomial degree in each variable, we can derive from it an explicit hard multilinear polynomial with only a polynomial deterioration in the hardness parameters. More precisely, replacing $x_i^{r}$, for $r > 1$ with $y_{i_0}^{r_0}\cdot \ldots \cdot y_{i_s}^{r_s}$. where $(r_0\ldots r_s)$ is the binary representation of $r$, gives a new multilinear polynomial in a slightly larger number of variables. This polynomial is at least as hard as the original polynomial which can be recovered
from it by the substitution $y_{ij} = x_i^{2^j}$. 
\end {remark}
\fi
As discussed earlier,~\autoref{thm:main thm} is closely related to the main result in~\cite{DSY09}. We now discuss their similarities and differences. 
\paragraph*{Comparison with~\cite{DSY09}. } 
\begin{itemize}
\item {\bf Degree constraint on the hard polynomial. } While~\autoref{thm:main thm} requires that the hard polynomial on $m$ variables has degree at most $O(\log^2 m/\log^2\log m)$, Dvir et al.~\cite{DSY09} did not have a similar constraint. 
\item {\bf Individual degree constraint for PIT. }In~\cite{DSY09}, the authors get PIT for low depth circuits with bounded individual degree, whereas our~\autoref{thm:main thm} does not make any assumptions on individual degrees in this context.
\end{itemize}
The key technical challenge for extending the known hardness-randomness tradeoffs for general circuits~\cite{ki03} to restricted circuit classes like formulas or bounded depth circuits is the following question :  Let $P(\vx, y)\in \F[\vx, y]$ be a polynomial of degree $r$ and let $f\in \F[\vx]$ be a polynomial of degree $d$ such that $P(\vx, f) \equiv 0$. Assuming $P$ can be computed by a low depth circuit (or arithmetic formula) of size at most $s$, can $f$ be computed by a low depth circuit (or arithmetic formula)  of size at most $\poly(s, n, d, r)$?

In~\cite{DSY09}, the authors partially answer this question by showing that the polynomial $f$ can be computed by a low depth circuit of size at most $\poly(s, r, d^{\deg_y(P)})$. Thus, for the case of polynomials $P$ which have small individual degree with respect to $y$, they answer the question in affirmative. 

Our main technical observation is the following result, which gives an upper bound on the \emph{low depth} circuit complexity of roots of low degree of a multivariate polynomial which has a small low depth circuit.
\begin{theorem}\label{thm:root ckt ub}
Let $\F$ be a field of characteristic zero. Let $P \in \F[\vx, y]$ be a polynomial of degree at most $r$ in $n+1$ variables that can be computed by an arithmetic circuit of size $s$ of depth at most $\Delta$. Let $f\in \F[\vx]$ be a polynomial of degree at most $d$ such that
\[
P(\vx, f) = 0 \, .
\] 
Then, $f$ can be computed by a circuit of depth at most $\Delta + 3$ and size at most $O((srn)^{10}d^{O(\sqrt{d})})$. 
\end{theorem}
We end this section with a short discussion on the \emph{low degree} condition in the hypothesis of~\autoref{thm:main thm}. 
\subsubsection{The Low Degree Condition}
The \emph{low degree} condition in the hypothesis of~\autoref{thm:main thm} appears to be extremely restrictive. It is natural to wonder if the question of proving super-polynomial lower bounds for constant depth circuits for an explicit polynomial family of \emph{low degree} much harder than the question of proving super-polynomial lower bound for constant depth circuits for an explicit polynomial family of potentially larger degree \footnote{In general, the degree only has to be upper bounded by a polynomial function in the number of variables.}? Currently, we do not even know quadratic lower bounds for arithmetic circuits of constant depth, and so, perhaps we are quite far from understanding this question. 

It is, however, easy to see that some of the known lower bounds for low depth circuits carries over to the low degree regime. For instance, the proofs of super-polynomial lower bounds for homogeneous depth-$3$ circuits by Nisan and Wigderson~\cite{nw1997}, super-polynomial lower bounds for  homogeneous depth-$4$ circuits based on the idea of shifted partial derivatives (see for example,~\cite{GKKS14, KSS13, FLMS13, KS14}) and super-linear lower bound due to Raz~\cite{Raz10a} do not require the degree of the hard function to be large. 

There are some known exceptions to this. For instance, lower bounds for homogeneous depth-$5$ circuits over finite fields due to Kumar and Saptharishi~\cite{KumarSapth15} are of the form $2^{\Omega(\sqrt{d})}$ and become trivial if $d < \log^2 n$. Another result which distinguishes the low degree and high degree regime is a separation between homogeneous depth-$5$ and homogeneous depth-$4$ circuit~\cite{KumarSapth15} which is only known to be true in the low degree regime (degree less than $\log^2 n$).

Another result of relevance is a result of Raz~\cite{raz10}, which shows that constructing an explicit family of tensors $T_n:[n]^d \rightarrow \F$, of rank at least $n^{d(1-o(1))}$ implies super-polynomial  lower bound for arithmetic formulas, provided $d \leq O(\log n/\log\log n)$. As far as we know, we do not know of such connections in the regime of high degree. 

One prominent family of lower bound results which do not seem to  generalize to this low degree regime are the super-polynomial lower bounds for multilinear formulas~\cite{Raz06}, and multilinear constant depth circuits~\cite{raz-yehudayoff}. In fact, the results in~\cite{raz10} show that super-polynomial lower bounds for set multilinear formulas for polynomials of degree at most $O(\log n/\log\log n)$ implies super-polynomial lower bounds for general arithmetic formulas. 

In the context of polynomial factorization, low degree factors of polynomials with small circuits have been considered before. For instance, Forbes~\cite{Forbes15} gave a quasi-polynomial time deterministic algorithm to test if a given polynomial of constant degree divides a given sparse polynomial. Extending this result to even testing if a given sparse polynomial divides another given sparse polynomial remains an open problem.

\subsection{Factors of Polynomials in $\VNP$}
We start by formally defining the complexity class $\VNP$. 
\begin{definition}[VNP]\label{def:VNP}
A family of polynomials $\{f_n\}$ over a field $\F$ is said to be in the class $\VNP_{\F}$ if there exist polynomially bounded functions $k, w, v : \N\rightarrow \N$ and a family $\{g_n\}$ in $\VP_{\F}$ such that for every sufficiently large $n \in \N$, 
\[
f_n(x_1, x_2, \ldots, x_{k(n)}) = \sum_{\vy \in \{0,1\}^{w(n)}} g_{v(n)}\left(x_1, x_2, \ldots, x_{k(n)}, y_1, y_2, \ldots, y_{w(n)}\right) \, .
\]
\end{definition}
We refer to the $y$ variables in the definition above as auxiliary variables, and the polynomial family $g_n$ as the family of verifier polynomials. 
Essentially, $\VNP$ can be thought of as the algebraic analog of $\NP$, and understanding if $\VNP$ is different from $\VP$ is the algebraic analog of the famous $\P$ vs $\NP$ question. As discussed earlier in this section, Kaltofen's closure result for $\VP$ does not seem to immediately extend to $\VNP$, and whether or not the factors of polynomials in $\VNP$ are in $\VNP$ was an open question. In~\cite{bur00}, B{\"u}rgisser conjectured the following. 
\begin{conjecture}[Conj. 2.1 in~\cite{bur00}]
The class $\VNP$ is closed under taking factors.
\end{conjecture}
As a direct application of our proof of~\autoref{thm:factor ckt ub-intro},  we confirm this conjecture over fields of characteristic zero or sufficiently large. We obtain a simple proof of the following statement. 
\begin{theorem}[Informal]
The class $\VNP$ is closed under taking factors.
\end{theorem}
The main technical statement which immediately gives us  this closure result is the following theorem.
%\begin{theorem}[Closure of $\VNP$ under Factoring]\label{thm:vnp closure-final}
%Let $\{f_n\}$ be a family of polynomials in $\VNP$ and let $\{g_n\}$ be a family of polynomials such that for every $n$, $g_n$ divides $f_n$. Then, $\{g_n\}$ is also in $\VNP$.
%\end{theorem}
%\autoref{thm:vnp closure-final} follows from the following more precise result. 
\begin{theorem}\label{thm:vnp closure-intro}
Let $\F$ be a field of characteristic zero. Let $P(\vx)$  be a polynomial of degree $r$ over $\F$, and let $Q(\vx, \vy)$ be a polynomial in $n + m$ variables such that
\[
P(\vx) = \sum_{\vy \in \{0,1\}^m} Q(\vx, \vy) \, , \text{ and }
\]
$Q$ can be computed by a circuit of size $s$. Let $f$ be any irreducible factor of $P$ of degree $d$. Then, there exists an $m' \leq \poly(s, r, d, n, m)$ and polynomial $h(x_1, x_2, \ldots, x_n, z_1, z_2, \ldots, z_{m'})$ where $h(\vx, \vz)$ can be computed by a circuit of size at most $s' \leq \poly(s, r, d, n, m)$ such that
\[
f(\vx) = \sum_{\vz \in \{0,1\}^{m'}} h(\vx, \vz) \, .
\]
\end{theorem}
We remark that in the proof of the above theorem, our techniques can be replaced by analogous statements from~\cite{DSY09, O16}. Although this is a simple observation, this does not appear to have been noticed prior to this work. The best upper bound on the complexity of factors of polynomials in $\VNP$ in prior  work is a recent result of Dutta, Saxena, Sinhababu~\cite{DSS17}, who showed a bound of $\poly(n, r, s, m, d^{O(\log d)})$ on the number of auxiliary variables and the circuit complexity of verifier polynomials $h$. 

As an easy consequence of our proofs, we also obtain another (slightly different) proof of the following result of Dutta et al~\cite{DSS17}. 
\begin{theorem}[\cite{DSS17}]
Let $P(\vx)$ be a polynomial of degree $r$ in $n$ variables which can be computed by an arithmetic formula (resp. algebraic branching program) of size at most $s$, and let $f(\vx)$ be a factor of $P$ of degree at most $d$. Then, $f(\vx)$ can be computed by an arithmetic formula (resp. algebraic branching program) of size at most $\poly(s, r, n, d^{O(\log d)})$. 
\end{theorem}

%\begin{remark}\label{rmk:depth tradeoff}
%The bound $d \leq \log^2m/\log^2\log m$ can be relaxed to $d \leq \log^k m/\log^k\log m$ for any positive integer $k$, but we would need lower bounds for depth $\Delta + 2k+2$ to be able to do PIT for depth $\Delta$ circuits. We point this difference out in the proof of~\autoref{thm:root ckt ub}, but do not dwell further on this. 
%\end{remark}

\section{Proof Overview}\label{sec:overview}
The key technical ingredients of our results in this paper is~\autoref{thm:root ckt ub}. We start by describing the main steps in its proof. 
\paragraph*{Proof sketch of~\autoref{thm:root ckt ub}. }
Our proof of~\autoref{thm:root ckt ub}  follows the outline of the proof of the analogous theorem about the structure of roots in~\cite{DSY09}. We now outline the main steps, and point out the differences between the proofs. The first step in the proof is to show that one can use the standard Hensel Lifting to iteratively obtain better approximations of the root $f$ given a circuit for $P(\vx, y)$. More formally, in the $k^{th}$ step, we start with a polynomial $h_k$  which agrees with $f$ on all monomials of degree at most $k$, and use it to obtain a polynomial $h_{k+1}$ which agrees with $f$ on all monomials of degree at most $k+1$. Moreover, the proof shows that if $h_k$ has a small circuit, then $h_{k+1}$ has a circuit which is only slightly larger than that of $h_k$. This iterative process starts with the constant term of $f$, which trivially has a small circuit. Thus, after $d$ iterations, we have a polynomial $h_d$ such that the root $f$ is the sum of the homogeneous components of $h_d$ of degree at most $d$. This lifting step is exactly the same as that in~\cite{DSY09} or in some of the earlier works on polynomial factorization~\cite{B04}, and is formally stated in~\autoref{lem:vanilla hensel}. 

The key insight of Dvir et al.~\cite{DSY09} was that if $\deg_y(P) = t$, and $C_0(\vx), C_1(\vx), \ldots, C_t(\vx)$ are polynomials such that $P(\vx, y) = \sum_{i = 1}^t C_i(\vx)y^t$, then for every $k \in \set{0, 1, \ldots, d}$, we have a polynomial $B_k$ of degree at most $k$ such that \[
h_k(\vx) = B_k(C_0(\vx), C_1(\vx), \ldots, C_t(\vx)) \, .
\]
Now, consider the case when $t << n$ (for instance $t = O(1)$). It follows from standard interpolation results for low depth circuits (see~\autoref{lem:interpolation}) that each of the polynomials $C_i(\vx)$ has a circuit of size $O(sr)$ and depth $\Delta$ since $P$ has a polynomial of size $s$ and depth $\Delta$. Thus, $h_d(\vx)$ can be written as a sum of at most $\binom{d+t}{t} = O(d^t)$ monomials if we treat each $C_i$ as a formal variable. Plugging in the small depth $\Delta$ circuits for each $C_i$, and standard interpolation (\autoref{lem:interpolation}), it follows that $f$ has a circuit of size $\poly(s, n, d^t)$ of depth $\Delta + O(1)$. 

Observe that this size bound of $\poly(s, n, d^t)$ is small only when $t$ is small. For instance, when $t > n$, this bound becomes trivial. Our key observation is that independently of $t$, there is a set of $d+1$ polynomials $g_0(\vx), g_1(\vx), \ldots, g_d(\vx)$ of degree at most $d$, and polynomials $A_0, A_1, \ldots, A_k$ on $d+1$ variables such that for every $k \in \set{0, 1, \ldots, d}$,
\[
h_k(\vx) = A_k(g_0(\vx), g_1(\vx), \ldots, g_d(\vx)) \, .
\]
Moreover, for every $k$, $A_k$ has degree at most $k$ and is computable by a circuit of size at most $O(d^3)$.  Also, each of these generators $g_i$ can be computed by a circuit of size $\poly(s, r)$ and depth $\Delta$. Thus, expressing $A_d(z_0, z_1, \ldots, z_d)$ as a sum of monomials, and then composing this representation with the circuits for $g_0, g_1, \ldots, g_d$ would give us a circuit of size $\poly(s, n, r, d, 4^d)$ of depth $\Delta + O(1)$. To get a sub-exponential dependence on $d$ in the size, we do not write $A_d(z_0, z_1, \ldots, z_d)$ as $\sum\prod$ circuit of size $O(4^d)$, but instead express it as a $\sum\prod\sum$ circuit of size at most $d^{O(\sqrt{d})}$, using the depth reduction result of~\cite{gkks13b}\footnote{See~\autoref{thm:depth-3 chasm} for a formal statement of this result.}. %For this to work, it is important that $A_d(z_0, z_1, \ldots, z_d)$ has a circuit of size $\poly(d)$, as opposed to a circuit of size $\poly(n)$, which would have given us a bound of $n^{O(\sqrt{d})}$ on the size of the depth-$3$ circuit. 

One point to note is that  just from Kaltofen's result~\cite{k89}, it follows that $f$ has an arithmetic circuit\footnote{Of potentially very large depth.} of size $\poly(n)$. Thus, from~\autoref{thm:depth-3 chasm}, it follows that $f$ has a circuit of depth-$3$ of size at most $n^{O(\sqrt{d})}$. The key advantage of~\autoref{thm:root ckt ub} over this bound is  that the exponential term is $d^{O(\sqrt{d})}$ and not of the form $n^{d^{\epsilon}}$. For $d \leq \log^2 n/\log^2\log n$, $d^{O(\sqrt{d})}$ is bounded by a polynomial in $n$ and so the final bound is at most $\poly(n)$. 
\paragraph*{Proof sketch of~\autoref{thm:factor ckt ub-intro}. }
To get~\autoref{thm:factor ckt ub-intro} from~\autoref{thm:root ckt ub}, we also have to upper bound the complexity of factors which are not of the form $y-f(\vx)$, i.e. are  non-linear in every variable. This involves the use of some standard techniques in this area.  We first preprocess $P$ such that it is monic in $y$,  and then we work over the algebraic closure of the field $\F[\vx]$, and view $P$ as a univariate in $y$ over this field. We then use~\autoref{lem:vanilla hensel} to approximate these roots by polynomials, and eventually combine them using~\autoref{lem:combining approx roots} from~\cite{O16} to obtain the factor $f$. 
We get bounds on the circuit size and depth of the factor $f$ by keeping tab on the growth of these parameters in each step of the outlined algorithm. 

\paragraph*{Proof sketch of~\autoref{thm:main thm}. }
\autoref{thm:root ckt ub}, when combined with the standard machinery of Nisan-Wigderson designs immediately yields~\autoref{thm:main thm}.

\paragraph*{Proof sketch of~\autoref{thm:vnp closure-intro}. }
For the proof of~\autoref{thm:vnp closure-intro}, we follow the same outline as above to conclude that every factor $f$ of a polynomial $P = \sum_{\vy \in \{0,1\}^{m}} {Q}(\vx, \vy)$ can be written as
\[
f(\vx) = \h_{\leq d}\left[B(g_0(\vx), g_1(\vx), \ldots, g_d(\vx)) \right]\, .
\] 
where $B$ has a circuit of size $\poly(d)$ degree at most $d$ and each polynomial $g_i$ can be written as $\sum_{\vy \in \{0,1\}^{m'}} \tilde{Q}_i(\vx, \vy)$, where the number of auxiliary variables $m'$ and the circuit size of $Q$ are each less than $\poly(s, n, m, d, r)$, where $s$ is the circuit size of $Q$, $r$ is the degree of $P$. The proof follows from a result of Valiant~\cite{v82}, where he showed that compositions such as $B(g_0(\vx), g_1(\vx), \ldots, g_d(\vx))$ can be written in the form $\sum_{\vy \in \{0,1\}^{m'}} Q'(\vx, \vy)$
with $m''$ and the circuit complexity of $Q'$ being at most $\poly(s, n, m, d, r)$. 
%{\color{red}\bf CN}

Note that composing $B$ and $g_i$ into the above form is not straightforward since direct replacement of $g_0$ with $\tilde{Q}_i$ might not work\footnote{Consider the following toy example: Let $B$ be a multiplication gate with two inputs from the same sub-circuit $g_0(\vx)$, \textit{i.e.,} $B(g_0(\vx))=g_0(\vx)^2$. However, if we directly replace $g_0(\vx)$ with $\tilde{Q}_0$, we would get $\sum_{y\in\{0,1\}^{m'}}\tilde{Q}_i(\vx,\vy)^2$, which might not be $g_0(\vx)$.}. For completeness, we include a proof of this using the depth reduction results in~\cite{vsbr83} (See~\autoref{thm:vnp structure} and~\autoref{clm:composed poly} and the appendix for the proof.).

We remark that the proof outlines above bounds the complexity of the factor $f$ once at the end of the lifting, whereas  in~\cite{DSS17},  the authors prove an upper bound on the number of auxiliary variables and the circuit complexity of the verifier circuit for the approximation of the factor of $P$ at the end of each step of the lifting process. They show that in every step of lifting, these parameters grow only by a multiplicative factor of at most $d^2$, and  there are $O(\log d)$ steps of lifting in total, hence the total blow up of $d^{O(\log d)}$ in the process. In contrast, we get a polynomial upper bound on the blow up in the number of auxiliary variables, and the circuit size of the  verifier circuit for the factor $f$, by a one step analysis. 

Another crucial point to note is that~\autoref{thm:vnp closure-intro} also follows if in the approach outlined above, we replace our structure theorem for the structure of low degree factors by an analogous statement in~\cite{DSY09} and~\cite{O16}. This is because, the degree of the factor we are seeking and the depth of the circuit obtained for the factor do not play a critical role in this proof as long as they are not too large. Thus, closure of $\VNP$ under taking factors follows from the results known prior to this work, although as far as we know, this does not seem to have been noticed before.

%We remark that in the proof of the above theorem, our techniques can be replaced by analogous statements from~\cite{DSY09, O16}. Although this is a simple observation, this does not appear to have been noticed prior to this work.
%We save on this quasipolynomial blow up by directly proving an upper bound on the 

%\paragraph*{Organization. } The rest of the paper is organized as follows. In~\autoref{sec:prelims}, we set up some notations and preliminaries for our proofs. In~\autoref{sec:}

%{\color{red}\bf CN (This has been discussed in Section 2.3. Do you want to remove this or put Section 2.3 here?)}

\section{Preliminaries}\label{sec:prelims}
We start by setting up some notation and stating some basic definitions and results from prior work which will be used in our proofs. 
\subsection{Notations}
\begin{itemize}
\item We use boldface letters $\vx, \vy, \vz $ to denote tuples of variables. 
\item For  a function $s(n):\N \rightarrow \N$, we say that $s(n) \leq \poly(n)$, if there are constants $n_0, a \in \N$ such that $\forall n > n_0$, $s(n) \leq n^{a}$.
\item For a polynomial $P$, $\deg(P)$ denotes the total degree of $P$ and $\deg_y(P)$ denotes the degree of $P$ with respect to the variable $y$. 
%\item Throughout this paper, we state and prove our results when the underlying field $\F$ is the field of rational numbers $\Q$, even though all our results hold as long as the field is of sufficiently large or zero characteristic.
\item Let $P\in \F[\vx]$ be a polynomial of degree equal to $d$. For every $k \in \N$, $\h_k\left[P\right]$ denotes the homogeneous component of $P$ of degree $k$. Similarly, $\h_{\leq k}\left[P \right]$ is defined to be equal $\sum_{i = 0}^k \h_i[P]$. 

\item We say that a polynomial $f$ is a factor of a polynomial $P$ of multiplicity equal to $m$, if $f^m$ divides $P$, and $f^{m+1}$ does not divide $P$. 
\end{itemize}
\subsection{Arithmetic Circuits}
\begin{definition}[Arithmetic Circuits]
An arithmetic circuit $\Psi$ over a field $\F$ and variables $\vx = (x_1, x_2, \ldots, x_n)$ is a directed acyclic graph, with the gates of in-degree zero (called leaves) being labeled by elements in $\F$ and variables in $\vx$, and the internal nodes being labeled by $+$ (sum gates) or $\times$ (product gates). The vertices of out-degree zero in $\Psi$ are called output gates. The circuit $\Psi$ computes a polynomial in $\F[\vx]$ in a natural way : the leaves compute the polynomial equal to its label. A sum gate computes the polynomial equal to the sum of the polynomials computed at its children, while a product gate computes the polynomial equal to the product of the polynomials computed at its children. 
\end{definition}
For an arithmetic circuit $\Psi$, we use $\size(\Psi)$ to denote the number of wires in $\Psi$. The depth of $\Psi$ is the length of the longest path from any output gate to any input gate. Throughout this paper, we assume that all our circuits are layered with alternating layers of addition and multiplication gates. Moreover, we always assume that the top layer is of addition gates. For instance, a depth-$3$ circuit is of the form $\sps$ and a depth-$4$ circuit is of the form $\sps\prod$.
\subsection{Derivatives}
We start by defining derivatives of a polynomial. For the ease of presentation, we work with the notion of the slightly non-standard notion of \emph{Hasse} derivatives even though we work with fields of characteristic zero. 
\begin{definition}[Derivatives]\label{def:derivative}
Let $\F$ be any field and let $P(y) \in \F[y]$ be a polynomial. Then for every $k \in \N$, the partial derivative of $P$ of order $k$ with respect to $y$ denoted by $\frac{\partial^k P(y)}{\partial y^k}$ or $P^{(k)}(y)$ is defined as the coefficient of $z^k$ in the polynomial $P(y + z)$.  
\end{definition}
We also use $P'(y)$ and $P''(y)$ to denote the first and second order derivatives of $P$ respectively.
An immediate consequence of this definition is the following lemma. 
\begin{lemma}[Taylor's expansion]\label{lem:taylor}
Let $P(y) \in \F[y]$ be a polynomial of degree $d$. Then,
\[
P(y + z) = P(y) + z\cdot P'(y) + z^2 \cdot P^{(2)}(y) + \cdots + z^d\cdot P^{(d)}(y) \, .
\]
\end{lemma}

\subsection{Depth Reductions}
We will use the following depth reduction theorems as a blackbox for our proofs. 
\begin{theorem}[Depth reduction to depth-$2k$~\cite{av08, koiran, Tav15}]\label{thm:depth-2k chasm}
Let $k$ be a positive integer and $\F$ be any field. If $P(\vx) \in \F[\vx]$ is an $n$-variate polynomial of degree $d$ that be computed by an arithmetic circuit $\Psi$ of size at most $s$, then $P$ can be computed by a depth $2k$ circuit of size at most $(snd)^{O(d^{1/k})}$.
\end{theorem}
Invoked with $k = 2$ the above theorem gives a circuit of depth $4$ for the polynomial $P$ of size $s^{O(\sqrt{d})}$. The next depth reduction result gives a further reduction to depth-$3$, as long as the field is of characteristic zero, and will be useful for our proof. 
\begin{theorem}[Depth reduction to depth-3~\cite{gkks13b}]\label{thm:depth-3 chasm}
Let $\F$ be a field of characteristic zero. Let $P(\vx) \in \F[\vx]$ be an $n$-variate polynomial of degree $d$ that can be computed by an arithmetic circuit $\Psi$ of size at most $s$. Then, $P$ can be computed by a $\sps$ circuit of size at most $(snd)^{O(\sqrt{d})}$.
\end{theorem}
We will also need the following theorem which gives a formula upper bound for polynomials with small circuits. The theorem immediately follows from a classical depth reduction result of Valiant, Skyum, Berkowitz and Rackoff~\cite{vsbr83}.
\begin{theorem}[\cite{vsbr83}]\label{thm:vsbr original}
Let $P(\vx)$ be a polynomial of degree $d$ in $n$ variables which can be computed by a circuit $C$ of size $s$. Then, $P$ can also be computed by a homogeneous circuit $C'$ of size  $\poly(s,n,d)$, with the following properties. 
\begin{itemize}
\item Every product gate in $C'$ has fan-in at most $5$. 
\item For every product gate $g$ in $C'$, the degree of the polynomial computed by any child  of $g$ is at most half of the degree of the polynomial computed at $g$. 
\item $C'$ has alternating layers of sum and product gates, where the sum fan-ins can be unbounded. 
\end{itemize} 
\end{theorem} 

\begin{theorem}[\cite{vsbr83}]\label{thm:circuit to formulas}
Let $P(\vx)$ be a polynomial of degree $d$ in $n$ variables which can be computed by a circuit of size $s$. Then, $P$ can also be computed by a formula of size $(sn)^{O(\log d)}$. 
\end{theorem} 

\subsection{Explicit Polynomials}
\begin{definition}[\cite{DSY09}]\label{def:explicit}
Let $\{f_m\}$ be a family of multilinear polynomials such that $f_m \in \Q[x_1, x_2, \ldots, x_m]$ for every $m$. Then, the family $\{f_m\}$ is said to be explicit if the following two conditions hold.
\begin{itemize}
\item All the coefficients of $f_m$ have bit complexity polynomial in $m$. 
\item There is an algorithm which on input $m$ outputs the list of all $2^m$ coeffcients of $f_m$ in time $2^{O(m)}$.
\end{itemize}
\end{definition}
%As remarked in the introduction, for this paper (as also in~\cite{DSY09}), explicit polynomial families do not hav
\subsection{Extracting Homogeneous Components}
For our proofs, we will also rely on the following classical result of Strassen, which shows that 
if a polynomial $P$ has a small circuit, then all its low degree homogeneous components also have small circuits. 
\begin{theorem}[Homogenization]\label{thm:homog strassen}
Let $\F$ be any field, and let $\Psi \in \F[\vx]$ be an arithmetic circuit of size at most $s$.  Then, for every $k \in \N$, there is a homogeneous circuit $\Psi_k$ of formal degree at most $k$ and size at most $O(k^2s)$, such that 
\[
\Psi_k = \h_{k}\left[\Psi \right] \, .
\]
\end{theorem}
\autoref{thm:homog strassen} gives us a way of extracting homogeneous components of the polynomial computed by a given circuit. We also need the following related well known lemma, whose proof we briefly sketch.
\begin{lemma}[Interpolation]\label{lem:interpolation}
Let $\F$ be any field with at least $d+1$ elements. Let $P(\vx, y) \in \F[\vx, y]$ be a polynomial of degree at most $d$. Let $C_0(\vx), C_1(\vx), \ldots, C_d(\vx) \in \F[\vx]$ be polynomials such that $P(\vx, y) = \sum_{j = 0}^d y^d\cdot C_j(\vx)$. Then, if $P(\vx, y)$ has a circuit of size at most $s$ and depth at most $\Delta$, then for every $j \in \set{0, 1, \ldots, d}$, $C_j(\vx)$ has a circuit of size at most $O(sd)$ and depth $\Delta$.
\end{lemma}
\begin{proof}[Proof Sketch]
For the proof, we view $P$ as a univariate polynomial of degree $d$ in $y$ with coefficients in the ring $\F[\vx]$. Thus, each $C_j$ can be written as an appropriate linear combination of $P(\vx, \alpha_0)$, $P(\vx, \alpha_1)$, $\ldots$, $P(\vx, \alpha_d)$, where $\alpha_0, \alpha_1, \ldots, \alpha_d$ are distinct elements of the field $\F$. Observe that for every $\alpha \in \F$, $P(\vx, \alpha)$ has a circuit of size at most $s$ and depth $\Delta$. To compute $C_j$, we have to take an appropriate linear combination of these circuits, but the linear combination can be absorbed in the top sum gate, and hence this process does not incur an increase in depth, while the size grows by a factor of at most $d$.
\end{proof}
The following corollary of~\autoref{lem:interpolation} would also be useful for us. The proof follows immediately from the proof of~\autoref{lem:interpolation}.  
\begin{lemma}[Interpolation for Formulas]\label{lem:interpolation formulas}
Let $\F$ be any field with at least $d+1$ elements. Let $P(\vx, y) \in \F[\vx, y]$ be a polynomial of degree at most $d$. Let $C_0(\vx), C_1(\vx), \ldots, C_d(\vx) \in \F[\vx]$ be polynomials such that $P(\vx, y) = \sum_{j = 0}^d y^d\cdot C_j(\vx)$. Then, if $P(\vx, y)$ has a \emph{formula} of size at most $s$, then for every $j \in \set{0, 1, \ldots, d}$, $C_j(\vx)$ has a \emph{formula} of size at most $O(sd)$.
\end{lemma}
\subsection{Hitting Sets}
\begin{definition}
A set of points ${\cal P}$ is said to be a hitting set for a class $\cal C$ of circuits, if for every $C \in {\cal C}$ which is not identically zero, there is an $\va \in {\cal P}$ such that $C(\va) \neq 0$. 
\end{definition}

Clearly, deterministic and efficient construction of a hitting set of small size for a class ${\cal C}$ of circuits immediately implies a deterministic PIT algorithm for ${\cal C}$. PIT algorithms designed in this way are also \emph{blackbox}, in the sense that they do not have to look inside into the wiring of the circuit to decide if it computes a polynomial which is identically zero. The PIT algorithms in this paper are all blackbox in this sense. 

\subsection{Nisan-Wigderson Designs}
We state the following well known result of Nisan and Wigderson~\cite{nw94} on the explicit construction of combinatorial designs. 
\begin{theorem}[\cite{nw94}]\label{thm:nw design}
Let $n, m$ be positive integers such that $n < 2^m$. Then, there is a family of subsets $S_1, S_2, \ldots, S_n \subseteq [\ell]$ with the following properties. 
\begin{itemize}
\item For each $i \in [n]$, $\abs{S_i} = m$.
\item For each $i, j \in [n]$, such that $i \neq j$, $\abs{S_i \cap S_j} \leq \log n$.
\item $\ell = O(\frac{m^2}{\log n})$.
\end{itemize}
Moreover, such a family of sets can be constructed via a deterministic algorithm in time $\poly(n, 2^{\ell})$.
\end{theorem}

\subsection{Schwartz-Zippel Lemma}
We now state the well known Schwartz-Zippel lemma.
\begin{lemma}[Schwartz-Zippel]\label{lem:sz lemma}
Let $\F$ be a field, and let $P\in \F[\vx]$ be a non-zero polynomial of degree (at most) $d$ in $n$ variables. Then, for any finite set $S\subset {\F}$ we have
$$|\{ \va \in S^n: P(\va) = 0\} | \le d {|S|}^{n-1}.$$
\end{lemma}

In particular, if $|S| \ge d+1$, then there exists some $\va \in S^n$ satisfying $P(\va) \ne 0$. 
This gives us a brute force deterministic algorithm, running in time $(d+1)^n$, to test if an arithmetic circuit computing a polynomial of degree at most $d$ in $n$ variables is identically zero.

\section{Low Degree Roots of Polynomials with Shallow Circuits}
In this section, we prove~\autoref{thm:root ckt ub}, which is also our main technical observation. We start with the following lemma, which gives us a way of \emph{approximating} the root of a polynomial to higher and higher accuracy, in an iterative manner. The lemma is a standard example of Hensel Lifting, which appears in many of prior works in this area including~\cite{DSY09}. The statement and the proof below, are from the work of Dvir et al~\cite{DSY09}.
\begin{lemma}[Hensel Lifting~\cite{DSY09}]\label{lem:vanilla hensel}
Let $P \in \F[\vx, y]$ and $f\in \F[\vx]$ be polynomials such that $P(\vx, f) = 0$ and $\h_{0}\left[\frac{\partial P}{\partial y}\left(\vx, f(\vx) \right)\right] = \delta \neq 0$. 
Let $i\in \set{1, 2, \ldots, \deg(f)}$ be any number. If $h\in \F[\vx]$ is a polynomial such that $\h_{\leq i-1}[f] = \h_{\leq i-1}[h]$, then 
\[ 
\h_{\leq i}\left[f\right] =\h_{\leq i}\left[h - \frac{ P(\vx, h)}{\delta} \right]\, .
\]
\end{lemma}
\begin{proof}
%We assume that we have homogeneous components of $f$ of degree $0, 1, \ldots, i-1$, and describe a recipe to compute the homogeneous component of $f$ of degree equal to $i$ from this information. We first need some notation.
For the rest of the proof, we think of $P(\vx, y)$ as an element of $\F[\vx][y]$. Henceforth, we drop the variables $\vx$ everywhere, and think of $P$ as a univariate in $y$. Thus, $P(y) = P(\vx, y)$. For brevity, we denote $\h_j[f]$ by $f_j$ for every $j \in \N$.

From the hypothesis, we know that $P(f) = 0$. Therefore,  $\h_{\leq i}(P(f)) = \h_{\leq i-1}\left[P(f) \right] =  0$. Moreover, since $\h_{\leq i-1}[h] = \h_{\leq i-1}[f]$, we get that $\h_{\leq i-1}\left[P(f) \right] =  \h_{\leq i-1}\left[P(h) \right] =  0$. So, we have 
\begin{align*}
0 &= \h_{\leq i}\left[P(f)\right] \\
%&\ \h_{\leq i}\left[P\left(\sum_{j = 0}^d f_j\right)\right] \\
&= \h_{\leq i}\left[P\left(h + (f_i-h_i)\right)\right]\\
%&\equiv \h_{\leq i}\left[P\left((f_i-h_i) + \sum_{j = 0}^{i-1} f_j\right)\right]\\
\end{align*}
We first observe that if $f_i = h_i$, then $\h_{\leq i}[P(h)] = 0$, and the lemma is trivially true. So, for the rest of the argument, we assume that $f_i-h_i \neq 0$. 
Now, by using~\autoref{lem:taylor}, we get the following equality.
\begin{align*}
0 &= \h_{\leq i}\left[P(h) + P'(h)\cdot (f_i-h_i) +  P''(h) \cdot (f_i-h_i)^2 + \ldots +  P^{(r)}(h)\cdot (f_i-h_i)^r\right]\\
&= \h_{\leq i}\left[ P(h)\right] + \h_{\leq i}\left[ P'(h)\cdot (f_i-h_i)  \right] + \ldots +  \h_{\leq i}\left[P^{(r)}(h)\cdot (f_i-h_i)^r \right] \\
\end{align*}
Here, $r$ denotes the degree of $P$. Since $f_i-h_i$ is non-zero, and every monomial in $f_i-h_i$ has degree equal to $i$, any term in the above summand which is divisible by $(f_i-h_i)^2$ does not contribute any monomial of degree at most $i$. Thus, we have the following.
\begin{align*}
0 &= \h_{\leq i}\left[ P(h)\right] + \h_{\leq i}\left[ P'(h)\cdot (f_i-h_i)  \right]\\
 &= \h_{\leq i}\left[ P(h)\right] + \h_{0}\left[P'(h)\right]\cdot (f_i-h_i) \, .  
\end{align*}
Now, we know that  $\h_{0}\left[P'(h))\right] = \h_{0}\left[P'(f)\right] = \delta \neq 0$.
Thus, 
\[
f_i = h_i - \frac{\h_{i}\left[ P(h)\right]}{\delta} \, .
\]
Since $\h_{\leq i-1}[P(h)]$ is identically zero, we get, 
\[
\h_{\leq i}\left[f \right] = \h_{\leq i}\left[ h - \frac{ P(h)}{\delta} \right] \, .
\]
\end{proof}
For our proof, we shall look at the structure of the outcome of the lifting operation in~\autoref{lem:vanilla hensel} more closely. Before proceeding further, we need the following crucial lemma. 
\begin{lemma}\label{lem:G properties}
Let $P(\vx, y) \in \F[\vx, y]$ be a polynomial of degree at most $r$, let $\alpha \in \F$ be a field element and $d \in \N$ be a positive integer. Let ${\cal G}_y'(P, \alpha, d)$ be the set of polynomials defined as follows.
\[
{\cal G}_y'(P, \alpha, d) = \set{\h_{\leq d}\left[ \frac{\partial^{j} P }{\partial y^j}\left(\vx, \alpha \right)\right] - \h_{0}\left[ \frac{\partial^{j} P }{\partial y^j}\left(\vx, \alpha \right)\right]: j \in \set{0, 1, 2, \ldots, d}}\, .
\]
Let ${\cal G}_y(P, \alpha, d)$ be the subset of ${\cal G}_y'(P, \alpha, d)$ consisting of all non-zero polynomials. Then, the following statements are true. 
\begin{itemize}
\item For every $g \in {\cal G}_y(P, \alpha, d)$, the degree of every non-zero monomial in $g$ is at least $1$ and at most $d$. 
\item $\abs{{\cal G}_y} \leq d + 1$.
\item If $P$ has a circuit of size at most $s$ and depth $\Delta$, then every $g \in {\cal G}_y(P, \alpha, d)$ has a circuit of size at most $O(sr^4)$ and depth $\Delta$.
\end{itemize}
\end{lemma}
%We remark that ${\cal G}$ is the subset of ${\cal G'}$ consisting of all polynomials of degree at least $1$. 
\begin{proof}
The first two items follow immediately from the definition of ${\cal G}_y(P, \alpha, d)$. We focus on the proof of the third item. Let $C_0(\vx), C_1(\vx), \ldots, C_{r}(\vx)$ be polynomials such that \[
P(\vx, y) = \sum_{i = 0}^{r} C_i(\vx) \cdot y^i \, .
\]
Now, for any $j \in \{0, 1, 2 , \ldots, d\}$, by~\autoref{def:derivative}, $\frac{\partial^{j} P }{\partial y^j}\left(\vx, y \right)$ is the coefficient of $z^j$ in $P(\vx, y + z)$. Moreover, 
\begin{align*}
P(\vx, y + z) &= \sum_{i = 0}^r C_i(\vx)\cdot (y + z)^i\, , \\
&= \sum_{i = 0}^r C_i(\vx)\cdot \left(\sum_{j = 0}^i \binom{i}{j} z^{j}y^{i-j}\right) \, , \\ 
&= \sum_{j=0}^r\left(\sum_{i=j}^r\binom{i}{j}C_i(\vx)\cdot y^{i-j}\right)\cdot z^j\, .
\end{align*}
Thus, for every $j \in \set{0, 1, \ldots, d}$, the coefficient of $z^j$ in $P(\vx, y + z)$ is given by $\sum_{i = j}^r \binom{i}{j}C_i(\vx)\cdot y^{i-j}$. From~\autoref{lem:interpolation}, we know that each $C_i(\vx)$ has a circuit of depth $\Delta$ and size at most $O(sr)$. Thus, we can obtain a circuit for $\binom{i}{j}C_i(\vx)\cdot y^{i-j}$ by adding an additional layer of $\times$ gates on top of the circuit for $C_i(\vx)$. This increases the size by an additive factor of $r$, and the depth by $1$. However, observe that this increase in depth is not necessary. Since, an expression of the form $y^i \cdot \left(\sum_a\prod_{b} Q_{a,b}\right)$ can be simplified to $ \sum_ay^i \cdot \left(\prod_{b} Q_{a,b}\right)$. Thus, the multiplication by $y^i$ can be absorbed in the product layer below the topmost layer of the circuits for $C_i(\vx)$, and this does not incur any additional increase in size. Thus, the polynomials $\frac{\partial^{j} P }{\partial y^j}\left(\vx, y \right)$, and hence $\frac{\partial^{j} P }{\partial y^j}\left(\vx, \alpha \right)$ have a circuit of size at most $O(sr^3)$ and depth at most $\Delta$. To compute the homogeneous components of these polynomials of degree at most $d$, we use~\autoref{lem:interpolation}. This increases the size by a factor of at most $O(r^2)$ while keeping the depth the same. 
\end{proof}
We now state our key technical observation.
\begin{lemma}\label{lem:key lem}
Let $P \in \F[\vx, y]$ and $f\in \F[\vx]$ be polynomials of degree $r$ and $d$ respectively such that $P(\vx, f) = 0$ and $\h_{0}\left[\frac{\partial P}{\partial y}\left(\vx, f(\vx) \right)\right] = \delta \neq 0$. Let the polynomials in the set ${\cal G}_y(P, \h_0[f], d)$ be denoted by $g_0, g_1, \ldots, g_d$. Then, for every $i\in \set{1, 2, \ldots, d}$, there is a polynomial $A_i(\vz)$  in $d+1$ variables such that the following are true. 
\begin{itemize}
\item $\h_{\leq i}\left[f\right] = \h_{\leq i}\left[A_i\left(g_0, g_1, \ldots, g_d\right)\right]$, and
\item $A_i(\vz)$ is computable by a circuit of size at most $10d^2i$.
%\item Given circuits for $g_1, g_2, \ldots, g_t$, we can obtain a circuit for $A_i\left(g_1, g_2, \ldots, g_t\right)$ by adding a further $O(d^2i)$ gates.
\end{itemize} 
\end{lemma}
This is an analog of the main technical lemma  in~\cite{DSY09}, which we state below. 
\begin{lemma}[\cite{DSY09}]\label{lem:dsy key lemma}
Let $P \in \F[\vx, y]$ and $f\in \F[\vx]$ be polynomials of degree $r$ and $d$ respectively such that $P(\vx, f) = 0$ and $\h_{0}\left[\frac{\partial P}{\partial y}\left(\vx, f(\vx) \right)\right] = \delta \neq 0$. Let $P(\vx, y) = \sum_{i = 0}^k C_i(\vx)\cdot y^i$. Then, for every $i\in \set{1, 2, \ldots, \deg(f)}$, there is a polynomial $A_i(\vz)$  in $k+1$ variables such that, 
 \[
 \h_{\leq i}\left[f\right] = \h_{\leq i}\left[A_i\left(C_0, C_1, \ldots, C_k\right)\right] \, .
 \]
\end{lemma}

The difference between these lemmas  is that in~\cite{DSY09}, it is shown that there is a set of polynomials of size at most $\deg_y(P) + 1$ which \emph{generate} every homogeneous component of the root $f$. Thus, in the regime of bounded individual degree, the size of this generating set is very small. However,  when $\deg_y(P) \geq n$,~\autoref{lem:dsy key lemma} does not say anything non-trivial since $f$ can be trivially written as a polynomial in the $n$ original variables. In contrast,~\autoref{lem:key lem} continues to say something non-trivial, as long as $d << n$, regardless of the value of $\deg_y(P)$. We now proceed with the proof.  
\begin{proof}[Proof of~\autoref{lem:key lem}]
%For brevity, we follow the notation from the proof of~\autoref{lem:vanilla hensel} for this proof as well. 
For the rest of the proof, we think of $P(\vx, y)$ as an element of $\F[\vx][y]$. So, we drop the variables $\vx$ everywhere, and think of $P$ as a univariate in $y$. Thus, $P(y) = P(\vx, y)$. For brevity, we denote $\h_j[f]$ by $f_j$ for every $j \in \N$. We also use ${\cal G}_y$ for ${\cal G}_y(P, f_0, d)$.  The proof will be by induction on $i$ and crucially use~\autoref{lem:vanilla hensel}. 
\begin{itemize}
\item {\bf Base case. } We first prove the lemma for $i = 1$. We invoke~\autoref{lem:vanilla hensel} with $i = 1$ and $h = f_0$. We get that 
\[
\h_{\leq 1}[f] = \h_{\leq 1}\left[f_0- \frac{P(f_0)}{\delta} \right]\, .
\]
The proof follows by observing that $f_0,\delta$ are constants and $\h_{1}\left[ P(f_0)\right] = \h_1\left[g_0 \right]$ where $g_0=\h_{\leq d}\left[ P(f_0)\right] - \h_{0}\left[ P(f_0)\right] \in {\cal G}_y$.
\item {\bf Induction step.} We assume that the claim in the lemma holds up to homogeneous components of degree at most $i-1$, and argue that it holds for $\h_{\leq i}[f]$. We  invoke~\autoref{lem:vanilla hensel} with $h = A_{i-1}(g_0, g_1, \ldots, g_d)$, which exists by the induction hypothesis.
\[
\h_{\leq i}\left[f \right] = \h_{\leq i}\left[ h - \frac{ P(h)}{\delta} \right] \, .
\]
Recall that $\h_0(h) = \h_0(f)$. Thus, $h = f_0 + \tilde{h}$, where every monomial in $\tilde{h}$ has degree at least $1$. By~\autoref{lem:taylor}, 
\[
P(f_0 + \tilde{h}) = P(f_0) + P'(f_0)\cdot \tilde{h} + \cdots + P^{(r)}(f_0)\cdot \tilde{h}^r \, .
\]
Thus, as $\tilde{h}$ has degree at least $1$, we have
\begin{align*}
\h_{\leq i}\left[f \right] &=\h_{\leq i}\left[h - \frac{1}{\delta} \cdot \left(P(f_0) + P'(f_0)\cdot \tilde{h} + \cdots + P^{(r)}(f_0)\cdot \tilde{h}^r \right)\right]\, , \\
&= \h_{\leq i}\left[h - \frac{1}{\delta} \cdot \left(P(f_0) + P'(f_0)\cdot \tilde{h} + \cdots + P^{(i)}(f_0)\cdot \tilde{h}^i \right)\right]\, . \\
\end{align*}
Since we are only interested in $i \leq d$, the following equality is also true. 
\begin{align*}
\h_{\leq i}\left[f \right] &= \h_{\leq i}\left[h - \frac{1}{\delta} \cdot \left(\h_{\leq d}\left[P(f_0)\right] + \h_{\leq d}\left[P'(f_0)\right]\cdot \tilde{h} + \cdots + \h_{\leq d}\left[P^{(i)}(f_0)\right]\cdot \tilde{h}^i \right)\right]\, . \\
\end{align*}
Observe that for every $j \in \set{0, 1, \ldots, d}$, $\h_{\leq d}\left[P^{(j)}(f_0)\right]$ is an affine form in the elements of $\cal G$\footnote{In fact, they are an affine form in one variable.}. For every $j \in \set{0, 1, 2, \ldots, i}$, let $\ell_j(\vz)$ be an affine form such that $\ell_j(g_0, g_1, \ldots, g_d) = \h_{\leq d}\left[P^{(j)}(f_0)\right]$. Now, we define $A_i(\vz)$ as 
\[
A_i(\vz) \equiv A_{i-1}(\vz) - \frac{1}{\delta}\left(\ell_0(\vz) + \ell_1(\vz)\cdot (A_{i-1}(\vz) - f_0) + \cdots + \ell_i(\vz)\cdot (A_{i-1}(\vz) - f_0)^i \right)\, .
\]
The first item in the statement of the lemma is true, just by the definition of $A_i(\vz)$ above. We now argue about the circuit size of $A_i(\vz)$. Each affine form $\ell_i(\vz)$ can be computed by a circuit of size at most $O(d)$. Thus, given a circuit of $A_{i-1}(\vz)$, we can obtain a circuit for $A_i(\vz)$ by adding at most $10d^2$ additional gates. Thus,  $A_i(\vz)$ can be computed by a circuit of size at most $10d^2(i-1) + 10d^2 = 10d^2i$ gates.
\end{itemize}
\end{proof}
The following is an easy corollary of~\autoref{lem:key lem}. 
\begin{corollary}\label{cor:key lem unique}
Let $P(\vx, y) \in \F[\vx, y]$ be a polynomial of degree, and $\alpha \in \F$ be such that $P({\bf 0}, \alpha) = 0$, and $\frac{\partial P}{\partial y}({\bf 0}, \alpha) \neq 0$. Then, for every $k \in \N$, there is a \emph{unique} polynomial $h_k(\vx)$ such that $\deg(h) \leq k$, $h_k(\bf 0) = \alpha$, and $\h_{\leq k}\left[P(\vx, h_k(\vx)) \right] = 0$. 
Moreover, if the polynomials in the set ${\cal G}_y(P, \alpha, k)$ be denoted by $g_0, g_1, \ldots, g_k$. Then, there is a polynomial $A_k(\vz)$  in $k+1$ variables such that the following are true. 
\begin{itemize}
\item $h_k = \h_{\leq k}\left[A_k\left(g_0, g_1, \ldots, g_k\right)\right]$, and
\item $A_k(\vz)$ is computable by a circuit of size at most $10k^3$.
\end{itemize} 
\end{corollary} 

The lemma  initially starts with an $\alpha \in \F$ such that $\alpha$ is a root of multiplicity $1$ of $P({\bf 0}, y)$. And, starting from this $\alpha$, we can lift \emph{uniquely} to a polynomial $h_i$ which is an \emph{approximate} root of the polynomial $P$. This corollary will be useful later on in the paper, when we study the structure of factors of $P$ which are not linear in $y$. And, the uniqueness will be important for this. 

We are now ready to complete the proof of~\autoref{thm:root ckt ub}. 
\begin{proof}[Proof of~\autoref{thm:root ckt ub}]
The first step is to massage the circuit for $P$ so that the hypothesis of~\autoref{lem:key lem} holds. We will have to keep track of the size and depth blow ups incurred in the process. We begin by ensuring that $f$ is a root of multiplicity $1$ of some polynomial related to $P$. 
\paragraph*{Reducing multiplicity of the root $f$. }
Let $P(\vx, y) = \sum_{i = 0}^r y^iC_i(\vx)$. Let $m \geq 1$ be the multiplicity of $f$ as a root of $P(\vx, y)$. Thus, $\frac{\partial^{j} P}{\partial y^j}\left(\vx, f\right) = 0$ for $j \in \set{0, 1, 2, \ldots, m-1}$, but $\frac{{\partial^{m} P}}{{\partial y^m}}\left(\vx, f\right) \neq 0$. The idea is to just work with the polynomial $\tilde{P} = \frac{{\partial^{m-1} P}}{{\partial y^{m-1}}}\left(\vx, y\right)$ for the rest of the proof. Clearly, $f$ is a root of multiplicity exactly $1$ of $\tilde{P}$. We only need to ensure that $\tilde{P}$ can also be computed by a small low depth circuit. This follows from the proof of the third item in~\autoref{lem:G properties}, where we argued that $\frac{{\partial^{j} P}}{{\partial y^{j}}}\left(\vx, y\right)$ has  a depth $\Delta$ circuit of size $\poly(s,r)$.
%Observe that $\frac{{\partial^{m-1} P}}{{\partial y^{m-1}}}\left(\vx, f\right)$ can be written as
%an appropriate linear combination of summands of the form $y^jC_{j'}(\vx)$. From~\autoref{lem:interpolation}, we know that each $C_{j'}(\vx)$ has a circuit of size at most $O(sD)$ and depth $\Delta$. Thus, $\frac{{\partial^{m-1} P}}{{\partial y^{m-1}}}\left(\vx, f\right)$ has a circuit of size at most $O(sD^4)$ and depth at most $\Delta + 2$. We have to add a $\sum\prod$ layer at the top where the product gates involve multiplications of the form $y^jC_{j'}(\vx)$. However, note that this multiplication by $y^j$ can be pushed to the product layer at the second level in the circuits  for $C_j$, and thus there is no increase in depth at this step. So, for the rest of this argument, we will assume that $f$ is the root of multiplicity equal to $1$ of a polynomial $P(\vx, y)$ given by a depth $\Delta$ circuit of size at most $O(sD^4)$
\paragraph*{Translating the origin. }From the step above, we can assume without loss of generality that $\frac{\partial P}{\partial y}\left(\vx, f\right) \neq 0$. Thus, there is a point $\va \in \F^n$ such that $\frac{\partial P}{\partial y}\left(\va, f(\va)\right) \neq 0$. By translating the origin, we will assume that $\frac{\partial P}{\partial y}\left({\bf 0}, f(0)\right) \neq 0$. This increases the depth of the circuit by at most $1$, as it could involve replacing every variable $x_i$ by $x_i + a_i$, and the size by at most a factor $n$. %Also, observe that if the bottom layer of the circuit for $P$ is a layer of $+$ gates, then the depth remains the same.
\iffalse
\paragraph*{Circuits for polynomials in ${\cal G}_y$. }
As discussed earlier, every partial derivative of $P$ of the form $\frac{{\partial^{j} P}}{{\partial y^{j}}}\left(\vx, y\right)$ has a circuit of size at most $O(sD^4)$ and depth at most $\Delta$. Thus, for every $j$, $\frac{{\partial^{j} P}}{{\partial y^{j}}}\left(\vx, \h_0[f]\right)$ has a circuit of size at most $O(sD^4)$ and depth at most $\Delta$. To obtain a circuit for its homogeneous components, we again apply~\autoref{lem:interpolation}, while noting that the total degree of each of these polynomials is at most $D$. Thus, every polynomial in $P$ has a circuit of size at most $O(s^2D^5)$ and depth at most $\Delta + 1$.
\fi
\paragraph*{Degree of $A_d$. } From~\autoref{lem:key lem}, we know that the polynomial $A_d(\vz)$ has a circuit of size at most $O(d^3)$. To obtain a circuit for $f$, we first prune away all the homogeneous components of $A_d(\vz)$ of degree larger than $d$. Recall that by definition, for every polynomial $g_i \in {\cal G}_y$, every non-zero monomial in $g_i$ has degree at least $1$, and that $f = \h_{\leq d}\left[A_d(g_1, g_2, \ldots, g_d)\right]$. Thus, any monomial of degree strictly greater than $d$ in $A_d(\vz)$ contributes no monomial of degree at most $d$ in the variables $\vx$ in the composed polynomial $A_d(g_1, g_2, \ldots, g_d)$, and hence does not contribute anything to the computation of $f$. So, we can confine ourselves to working with the homogeneous components of $A_d(\vz)$ of degree at most $d$. 

By~\autoref{thm:homog strassen}, we know that given a circuit for $A_d(\vz)$, we can construct a circuit for $\h_{i}\left[A_d(\vz) \right]$ by increasing the size of the circuit by a multiplicative factor of at most $O(i^2)$. Thus, $\h_{\leq d}[A_d(\vz)]$ can be computed by a circuit of size at most $O(d^3)\times \size(A_d(\vz))$. Thus, for the rest of this argument, we will assume that $A_d(\vz)$ has a circuit of size at most $O(d^6)$ and degree at most $d$, and 
\[
f = \h_{\leq d}\left[A_d(g_1, g_2, \ldots, g_d)\right] \, .
\] 
\paragraph*{Circuit for $A_d(\vz)$ of small depth. }
Given that $A_d(\vz)$ has a circuit of size $O(d^6)$ and degree at most $d$, by~\autoref{thm:depth-3 chasm}, we know that $A_d(\vz)$ can be computed by a  $\sps$ circuit $\Psi$ of size at most $d^{O(\sqrt{d})}$\footnote{Instead of~\autoref{thm:depth-3 chasm}, one could use~\autoref{thm:depth-2k chasm} to get a better size bound than $d^{O(\sqrt{d})}$ at the cost of increasing its depth appropriately.}. Similar results follow from the application of~\autoref{thm:depth-2k chasm}. 

\paragraph*{Circuit for $f$ of small depth. }Composing the $\sps$ circuit $\Psi$ for $A_d(\vz)$ with the circuits of $g_1, g_2, \ldots, g_d \in {\cal G}_y$, we get a circuit $\Psi'$ with the following properties. 
\begin{itemize}
\item The size of $\Psi'$ is at most $(srn)^{10}\cdot d^{O(\sqrt{d})})$. 
\item The depth of $\Psi'$ is at most $\Delta + 3$. This follows by combining the bottom $\sum$ layer of the $\sps$ circuit for $A_d(\vz)$ with the top $\sum$ layer of the circuits for $g_i \in {\cal G}_y$.
\item The degree of $\Psi'$ is at most $d^2$. This is true because the degree of $A_d(\vz)$ is at most $d$ (as argued earlier in this proof), and the degree of every polynomial in ${\cal G}_y$ is at most $d$ (first item in~\autoref{lem:G properties}).
\item $f = \h_{\leq d}\left[\Psi'(\vx)\right]$.
\end{itemize}
To obtain a circuit for $f$, we apply~\autoref{lem:interpolation} to $\Psi'$. This increases the size of $\Psi'$ by a multiplicative factor of at most $O(d^2)$, while the depth remains the same. This completes the proof of the theorem.
\end{proof}

%\section{Open questions}

\section{From Roots to Arbitrary Factors}
In this section, we show that~\autoref{thm:root ckt ub} essentially generalizes to arbitrary factors, and not necessarily factors of the form $y-f(\vx)$, up to some loss in the size and depth parameters. The techniques for this generalization are quite standard and well known in this literature, and our presentation here follows the approach of Oliveira~\cite{O16}. We sketch the main steps towards obtaining circuits for arbitrary factors.  
\paragraph*{Making the Polynomial Monic in $y$. } Starting with an arbitrary polynomial $P(\vx, y)$, we first make sure that it is monic in $y$. We do this by taking an invertible linear transformation $x_i \rightarrow x_i + a_i\cdot y$, where the vector $\va$ is chosen randomly from some large enough grid. Indeed, assume that $\deg(P)= r$. Let us consider the homogeneous component of degree $r$ of $P(\vx, y)$.  Since $\h_{r}[P(\vx,y)]$ is homogeneous in $(\vx,y)$ of degree $r$, so $\h_r[P(\vx, y)] = P_r(\vx/y, 1)\cdot y^r$ for a polynomial $P_r$, implying that $P_r(\vx/y,1)$ is not the zero polynomial, so we can write
$$P(\vx + \va y, y) = P_r(\va, 1)y^r + \text{lower order terms (in }y)\, .$$

By~\autoref{lem:sz lemma}, there exists some $\va \in [r+1]^n$, with $P_r(\va,1) \ne 0$. Thus, in the inverted coordinate system, the leading coefficient of $P(\vx + \va y, y)$ (as a polynomial in $y$), is some non-zero element of the field $\F$, and, without loss of generality, we can take it to be $1$. 

If $P(\vx,y)$ is monic, then so are its factors. To see this, we first apply the following lemma due to Gauss, to deduce that its factors are elements in $\F[\vx,y]$. 
\begin{lemma}[Gauss' Lemma]\label{lem:Gauss}
	Let $R$ be a Unique Factorization Domain with a field of fractions $F$ and let $f(y) \in R[y]$. If $f(y)$ is reducible over $F[y]$, then $f$ is reducible over $R[y]$.
\end{lemma}
See Lemma 12.2 in~\cite{algcomp17} for a proof of this lemma. Now, we have the following simple observation.
\begin {observation}
Put $R=\F[x]$ be a Unique Factorization Domain. 
If $P \in R[y]$ is a monic polynomial in $y$, and $P=g \cdot h$, where $g,h \in R[y]$, then the leading coefficients of $g$ and $h$ in $y$ belong to $\F \setminus \{0\}$.
\end {observation}
\begin {proof} 
Write $g = \sum_{i=0}^{d} a_i(x) y^i$. and $f = \sum_{i=0}^{d'} b_i(x) y^i$, where $d = \deg_y(g), d' = \deg_y(h)$.
As $P$ is monic in $y$, its leading coefficient, as a polynomial in $y$ with coefficients in $R$, is $1$. Since $P=g\cdot h$, it follows that $1 = a_{d} \cdot b_{d'}$, implying that $a_d$ and $b_{d'}$ are invertible elements in $R$ by~\autoref{lem:Gauss}. Thus, $a_d$ and $b_{d'}$ are precisely the elements in $\F\setminus \{0\}$.
\end {proof}
Thus, for the rest of this section, we will assume that all the factors of $P(\vx, y)$ are also monic in $y$. 
\paragraph*{Working over the algebraic closure of $\F(\vx)$. }
As above, $P$ is monic in $y$ with $deg_y(P) = r$, that is,
$$P(\vx, y) = y^r + \sum_{i=0}^{r-1} P_r(\vx) y^i.$$
Assume that $P$ does not factor into linear factors in $y$, and that $f(\vx,y)$ is one of its factors, of degree $k$ in $y$. Since $P$ is monic in $y$, we know that $f$ must also be monic in $y$ by~\autoref{lem:Gauss}. Working over the algebraic closure of $\F(\vx)$ (that is, the field $\overline{\F(\vx)}$), we can factor $P$ (and $f$) into linear factors in $y$. The algebraic closure of $\F(\vx)$ is a complicated object, but we only need to think of elements of the closure as ``functions'' over the variables in $\vx$. Since $f$ divides $P$ , if
$$P(\vx,y) = \prod_{i=1}^r (y-\varphi_i(\vx))\, ,$$
Without loss of generality, assume the first $d$ of these $\varphi_i$ corresponds to roots of $f$, we have
$$f(\vx,y) = \prod_{i=1}^d (y-\varphi_i(\vx))\, .$$
We note that $\varphi_i(\vx)$ may not be polynomials in $\vx$\footnote{As shown in~\cite{DSS17}, $\varphi_i(\vx)$ could be power series in $\vx$.}. Still, the fact that they share some roots in the closure of $\F(\vx)$ gives us a way to approximate them, using Hensel's lifting, similar to~\autoref{lem:key lem}. For the rest of our argument, we first need to ensure some non-degeneracy conditions. 
\paragraph*{Reducing the multiplicity of $f$ in $P$. } We first make sure that $f$ is a factor of multiplicity $1$ of $P$; if $f$ is a factor of multiplicity $m > 1$, we can replace $P$ by $\tilde{P} = \frac{{\partial^{m-1} P}}{{\partial y^{m-1}}}\left(\vx, y\right)$. Clearly, $f$ is a factor of multiplicity exactly $1$ of $\tilde{P}$. Ensuring that $\tilde{P}$ can also be computed by a small low depth circuit, follows from the proof of the third item in~\autoref{lem:G properties}, where we argued that $\frac{{\partial^{j} P}}{{\partial y^{j}}}\left(\vx, y\right)$ has  a depth $\Delta$ circuit of size $O(sr^3)$.
So, for the rest of the proof, we will assume that $f$ is a factor of $P$ of multiplcity equal to $1$.
\paragraph*{Properly Separating Shifts. }
To proceed further, we want a shift in $\vx$ such that each factor has no repeating roots in $y$ and distinct factors share no root in $y$. This follows from the below lemma from~\cite{O16}, which we state without a proof. 
\begin{lemma}[Lem. 3.6 in~\cite{O16}]\label{lem:properly separating shifts}
Let $f(\vx, y), g(\vx, y) \in \F[\vx, y]$ be polynomials such that $\deg_y(f) \geq 1, \deg_y(g) \geq 1$, $f$ is irreducible and $f$ does not divide $g$. Then, there is a $\vc \in \overline{\F}^n$ such that 
\begin{itemize}
\item $f(\vc, y)$ is a polynomial with exactly $\deg_y(f)$ distinct roots in $\F$, and
\item $f(\vc, y)$ and $g(\vc, y)$ have no common roots. 
\end{itemize}
\end{lemma}
 Now, let us consider the polynomial $g = P/f$. Since $f$ is factor of multiplicity $1$ of $P$,  $f$ does not divide $g$. From~\autoref{lem:properly separating shifts}, we know that there is a $\vc\in \F^n$ such that $f(\vc, y)$ and $g(\vc, y)$ do not share a root, and all the roots of $f(\vc, y)$ are distinct. At the cost of increasing the depth of the circuit of $P$ by $1$, we can assume without loss of generality that $\vc$ is the origin. So, for the rest of the proof, we assume that $f({\bf 0}, y)$ has no repeating roots, and $f({\bf 0}, y)$ and $g({\bf 0}, y)$ share no common roots. Let $\alpha_1, \alpha_2, \ldots, \alpha_r$ be the roots of $P({\bf 0}, y)$ and let $\alpha_1, \alpha_2, \ldots, \alpha_d$ be the roots of $f({\bf 0}, y)$. 
\paragraph*{Approximating the roots of $P$. }
The goal of this step is to approximate the roots of $P$ by small degree polynomials with small circuits.
From the previous paragraph, we know that for $i \in [d]$, $P({\bf 0}, \alpha_i) = 0$ and $\frac{\partial P}{\partial y}({\bf 0}, \alpha_i) \neq 0$. 
Thus, from~\autoref{cor:key lem unique}, there are polynomials $q_1, q_2, \ldots, q_d$ of degree at most $d$ such that for every $i \in [d]$, there is a polynomial $A_{i,d}(\vz)$  in $d+1$ variables such that the following are true. 
\begin{itemize}
\item $q_i({\bf 0}) = \alpha$, and
\item $q_i = \h_{\leq d}\left[A_{i, d}\left(g_{i,0}, g_{i,1}, \ldots, g_{i, d}\right)\right]$, and
\item $A_{i, d}(\vz)$ is computable by a circuit of size at most $10d^3$.
\end{itemize} 
Here, for every $i \in [d]$, $g_{i,0}, g_{i,1}, \ldots, g_{i, d}$ are the polynomials in the set ${\cal G}_y(P, \alpha_i, d)$. Thus, we have degree $d$ polynomials, which are approximations of the roots of $P$, the constant terms of these polynomials agree with the roots of $f(\vx, 0)$ and these approximate roots have \emph{small} low depth circuits. Moreover, We will now combine these approximations to obtain a circuit for $f$. 
\paragraph*{Obtaining a Circuit for $f$. }
In this final step, we are going to obtain circuit for $f$ from the polynomials $q_1, q_2, \ldots, q_d$ in the previous step. The first observation is that the $q_1, q_2, \ldots, q_d$ are also approximate roots of the polynomial $f$. To see this, observe that by our choice, $\alpha_1, \alpha_2, \ldots, \alpha_d$ are distinct roots of $f({\bf 0}, y)$. Thus, for each $i \in [d]$, $f({\bf 0}, \alpha_i) = 0$ and $\frac{\partial f}{\partial y}({\bf 0}, \alpha_i) \neq 0$. Thus, by~\autoref{cor:key lem unique}, there are degree $d$ polynomials $\tilde{q}_1, \tilde{q}_2, \ldots, \tilde{q}_d$ of degree at most $d$ such that $\h_{\leq d}[f(\vx, \tilde{q}_i(\vx))] = 0$. Thus, we also have $\h_{\leq d}[P(\vx, \tilde{q}_i(\vx))] = 0$. So, by the uniqueness condition in~\autoref{cor:key lem unique}, we get that the set of polynomials $\{\tilde{q}_i : i \in [d]\}$ must be the same as $\{q_i : i \in [d]\}$. 

Next, to obtain a circuit for $f$, we now claim that 
\[
f(\vx, y) = \h_{\leq d}\left[\prod_{i = 1}^d \left(y-q_i(\vx) \right)\right] \, .
\]
The proof of this fact follows immediately from Lemma 5.4 in~\cite{O16}. We state a special case, which suffices for our application.
\begin{lemma}[Lem 5.4 in~\cite{O16}]\label{lem:combining approx roots}
Let $P(\vx, y)$ and $f(\vx, y)$ be polynomials of degree $r$ and $d$ respectively,  such that $P$ and $f$ are monic in $y$, $f$ is a factor of $P$ and all the roots of $f({\bf 0}, y)$ are distinct and roots of multiplicity exactly one of $P({\bf 0}, y)$. Let $\alpha_1, \alpha_2, \ldots, \alpha_d$ be the roots of $f({\bf 0}, y)$ and let $q_1, q_2, \ldots, q_d \in \F[\vx]$ be polynomials of degree at most $d$ such that for every $i\in [d]$, 
\begin{itemize}
\item $q_i(\bf 0) = \alpha_i$, 
\item $\h_{\leq d}[P(\vx, q_i(\vx))] = \h_{\leq d}[f(\vx, q_i(\vx))] = 0$. 
\end{itemize}
Then, 
\[
f = \h_{\leq d}\left[\prod_{i = 1}^d \left( y - q_i(\vx)\right)\right] \, .
\]
\end{lemma}
Thus, given the circuits for $q_i(\vx)$, we can obtain a circuit for $f(\vx, y)$ by increasing the depth by at most two (a product layer, and then a sum layer for interpolation), and size by a $\poly(d)$ factor. In summary, we have the following two statements. 
\begin{lemma}\label{lem:key lem gen}
Let $P \in \F[\vx, y]$ and $f\in \F[\vx, y]$ be polynomials of degree $r$ and $d$ respectively such that $P$ is monic in $y$ and $f$ is an irreducible factor of $P$. Then, there exist $\vc \in \F^n$, $\alpha_1, \alpha_2, \ldots, \alpha_d \in \F$  and a polynomial $B(\vz)$ of degree at most $d$ in $t = O(d^2)$ variables, such that the following are true. 
\begin{itemize}
\item $f(\vx + \vc, y) = \h_{\leq d}\left[B\left(g_0, g_1, \ldots, g_t\right)\right]$, where $g_1, g_2, \ldots, g_t$ are polynomials in the set \\
$\bigcup_{i = 1}^d {\cal G}_y\left(P(\vx + \vc , y), \alpha_i, d\right)$. 
\item $B(\vz)$ is computable by a circuit of size at most $\poly(d)$.
\end{itemize} 
\end{lemma}
\begin{theorem}\label{thm:factor ckt ub}
Let $P \in \F[\vx, y]$ be a polynomial of degree at most $r$ in $n+1$ variables that can be computed by an arithmetic circuit of size $s$ of depth at most $\Delta$. Let $f\in \F[\vx, y]$ be an irreducible  polynomial of degree at most $d$ such that $f$ divides $P$. 
Then, $f$ can be computed by a circuit of depth at most $\Delta + O(1)$ and size at most $O(\poly(s, r, n)\cdot d^{O(\sqrt{d})})$. 
\end{theorem}
 \section{Deterministic PIT for Low Depth Circuits from Hardness}
In this section, we use~\autoref{thm:root ckt ub} to show that given a family of polynomials which are hard for depth $\Delta$ circuits, we can do deterministic identity testing for $\Delta-5$ circuits in subexponential time. In short, the high-level strategy is to generate hitting set for low depth circuits from the hard polynomial combined with Nisan-Wigderson designs. Since the content of this part are very similar to the proofs of similar statements in~\cite{ki03} and~\cite{DSY09}, we only outline the differences in the proofs (if any), and refer the reader to~\cite{DSY09} for details. We start with the following lemma, which is the analog of Lemma 4.1 in~\cite{DSY09}.
\begin{lemma}[Analog of Lemma 4.1 in~\cite{DSY09}]\label{lem:hybrid}
Let $q(\vx) \in \F[\vx]$ be a (non-zero) polynomial of degree $D$ in $n$ variables, which can be computed by a circuit of size $s$ and depth $\Delta$. Let $m > \log n$ be an integer and let $S_1, S_2, \ldots, S_n \subseteq [\ell]$ be given by~\autoref{thm:nw design}, so that $\ell = O(m^2/\log n)$, $\abs{S_i} = m$, and $\abs{S_i \cap S_j} \leq \log n$. For  a multilinear polynomial $f\in \F[z_1, z_2, \ldots, z_m]$ of degree $d$, put 
\[
Q(\vy) = Q(y_1, y_2, \ldots, y_{\ell}) := q\left(f(\vy|_{S_1}), f(\vy|_{S_2}), \ldots,f(\vy|_{S_n})\right) \, . 
\]
If $Q(\vy) \equiv 0$, then $f(\vz)$ can be computed by an arithmetic circuit of size $O((snD)^{12} d^{O(\sqrt{d})})$ and depth at most $\Delta + 5$.
\end{lemma}
Note that the bound on the size of $f$ remains non-trivial as long as $d << m$, while the individual degree of $q$ is allowed to be unbounded, whereas the bound in~\cite{DSY09} becomes trivial once $\deg_y(q)$ is larger than $m$.  
\begin{proof}[Proof Sketch]
The proof is along the lines of the proof of Lemma 4.1 in~\cite{DSY09}. We now give a sketch of the details. We first define the hybrid polynomials $Q_0(\vx, \vy), Q_1(\vx, \vy), \ldots, Q_n(\vx, \vy)$ as follows. 
\[
Q_j(\vx, \vy) =  q\left(f(\vy|_{S_1}), f(\vy|_{S_2}), \ldots, f(\vy|_{S_j}), x_{j+1}, x_{j+2}, \ldots, x_n\right) \, .
\]
We know that $Q_0(\vx, \vy)$ is non-zero, whereas $Q_n(\vx, \vy)$ is identically zero. Thus, there is an $i \in \set{0, 1, \ldots, n}$ such that $Q_i(\vx, \vy) \not\equiv 0$ and $Q_{i+1}(\vx, \vy) \equiv 0$. We now fix the variables $x_{i+2}, x_{i+3}, \ldots, x_{n}$ and the variables $\set{y_j : j \notin S_{i+1}}$ to field constants while maintaining the non-zeroness of $Q_i$. This can be done via~\autoref{lem:sz lemma}. Thus, we have  a polynomial $\tilde{q}$ by fixing the aforementioned variables such that the following two conditions hold.
\[
 \tilde{q}\left(f(\vy|_{S_1\cap S_{i+1}}), f(\vy|_{S_2\cap S_{i+1}}), \ldots, f(\vy|_{S_i\cap S_{i+1}}), x_{i+1}\right)\not\equiv 0 \, .
\]  
\[
\tilde{q}\left(f(\vy|_{S_1\cap S_{i+1}}), f(\vy|_{S_2\cap S_{i+1}}), \ldots, f(\vy|_{S_i\cap S_{i+1}}), f(\vy|_{S_{i+1}})\right)\equiv 0 \, .
\]  
Let $A_0(\vy|_{S_{i+1}}, x_{i+1})$ denote the polynomial $\tilde{q}\left(f(\vy|_{S_1\cap S_{i+1}}), f(\vy|_{S_2\cap S_{i+1}}), \ldots, f(\vy|_{S_i\cap S_{i+1}}), x_{i+1}\right)$. 
The above two conditions imply that $f(\vy|_{S_{i+1}})$ is a root of the polynomial $A_0(\vy|_{S_{i+1}}, x_{i+1}) \in \F[\vy|_{S_{i+1}}][x_{i+1}]$, viewed as a polynomial in $x_{i+1}$. Moreover, $A_0(\vy|_{S_{i+1}}, x_{i+1})$ has a circuit of size at most $O(sn)$ and depth at most $\Delta + 2$. This follows from the fact that $f(\vy|_{S_1\cap S_{i+1}})$ is a \emph{multilinear} polynomial in $\log n$ variables, and can thus be computed by a $\sum\prod$ circuit of size at most $n$. We simply replace the variables $x_1, x_2, \ldots, x_i$ in the circuit for $q$ by these $\sum\prod$ circuits to obtain a circuit for $A_0$. The degree of $A_0$ is at most $D\log n$. Finally,~\autoref{thm:root ckt ub} implies that $f(\vy|_{S_{i+1}})$ can be computed by a circuit of size at most $O(\poly(s,n,D) d^{O(\sqrt{d})})$ and depth at most $\Delta + 5$, thus completing the proof. 
\end{proof}

We now sketch the proof of~\autoref{thm:main thm}. 
\begin{proof}[Proof Sketch]
	Once again, the proof follows the proof of Theorems 1 and 2 in~\cite{DSY09}. Let $\{f_m\}$ be a family of explicit multilinear polynomials such that $f_m$ has $m$ variables, degree $d \leq O\left(\left(\frac{\log m}{\log\log m}\right)^2\right)$, such that $f_m$ cannot be computed by a circuit of depth $\Delta$ and size $\poly(m)$.  Let $\epsilon \in (0,  0.49)$ be an arbitrary constant, and set $m:=n^\epsilon$, and $f=f_m$. 
	
	Given as input a circuit $C\in\F[\vx]$ of size $s$, depth $\Delta-5$ and degree $D$ on $n$ variables, let $q\in\F[\vx]$ be the polynomial computed by $C$. The goal here is to determine whether $q$ is nonzero. From the equivalence of black-box PIT and hitting set, it suffices to construct hitting set for circuit class of the above properties.
	\begin{itemize}
		\item We construct a design $S_1, S_2, \ldots, S_n \subseteq [\ell]$ using~\autoref{thm:nw design} where each set $S_i$ has size $m$, $\ell = O(m^2/\log n) \leq n^{2\epsilon} < n^{0.98}$ and $\abs{S_i \cap S_j} \leq \log n$. This can be done in deterministic time $2^{O(n^{2\epsilon})}$.
		\item We pick a subset $T$ of the field $\F$ of size $Dd + 1$ and evaluate the polynomial \linebreak
		$q\left(f(\vy|_{S_1}), f(\vy|_{S_2}), \ldots,f(\vy|_{S_n})\right)$ on all points of $T^{\ell}$. 
		$H=\{(f(\vy|_{S_1}),f(\vy|_{S_2}),\ldots,f(\vy|_{S_n}))\ |\ \vy\in T^{\ell} \}$ is then our candidate hitting set of size $(Dd + 1)^{\ell} = n^{O(n^{2\epsilon})} < n^{O(n^{0.98})}$. Note that the set can be constructed deterministically in time $m^d\cdot n^{O(n^{2\epsilon})} = n^{O(n^{2\epsilon})}$.
	\end{itemize}
	We now argue about the correctness, \textit{i.e.,} $q$ does not vanish on the hitting set if and only if $q$ is not identically zero. Observe that if the polynomial $q\left(f(\vy|_{S_1}), f(\vy|_{S_2}), \ldots,f(\vy|_{S_n})\right)$ is not identically zero, then it has degree at most $Dd$ and hence by~\autoref{lem:sz lemma}, $q$ does not vanish on the set $H$. Else, $q\left(f(\vy|_{S_1}), f(\vy|_{S_2}), \ldots,f(\vy|_{S_n})\right) \equiv 0$. But then, by~\autoref{lem:hybrid}, we get that $f$ can be computed by a circuit of depth $\Delta$ and size at most $O\left(\poly(s,n,D)d^{O(\sqrt{d})}\right)$. If $s, D$ are $\poly(n)$, then this bound is $\poly(m)$ which contradicts the assumed hardness of $f=f_m$ for circuits of depth $\Delta$. This shows that $H$ is a hitting set for the desired circuit  class and completes the proof.
\end{proof}

\section{Factors of Polynomials in VNP}

We now prove~\autoref{thm:vnp closure-intro}, which is restated below.
\begin{theorem}[\autoref{thm:vnp closure-intro} restated]\label{thm:vnp closure}
Let $P(\vx)$ be a polynomial of degree $r$ over $\F$, and let $Q(\vx, \vy)$ be a polynomial in $n + m$ variables such that
\[
P(\vx) = \sum_{\vy \in \{0,1\}^m} Q(\vx, \vy) \, , \text{ and }
\]
$Q$ can be computed by a circuit of size $s$. Let $f$ be an irreducible factor of $P$ of degree $d$. Then, there exists an $m' \leq \poly(s, r, d, n, m)$ and polynomial $h(\vx, z_1, z_2, \ldots, z_{m'})$, such that $h(\vx, \vz)$ can be computed by a circuit of size at most $s' \leq \poly(s, r, d, n, m)$ and 
\[
f(\vx) = \sum_{\vz \in \{0,1\}^{m'}} h(\vx, \vz) \, .
\]
\end{theorem}

%We note that~\autoref{thm:vnp closure} follows from combining~\autoref{lem:key lem gen} with the closure properties of $\VNP$ as described in~\cite{v82} (see Theorem 2.19 in~\cite{bur00}). However, we include a proof using~\autoref{lem:key lem gen} and the depth reduction result in~\cite{vsbr83}, as it is fairly self contained and is in the spirit of our proof of~\autoref{thm:factor ckt ub}.

For our proof, we use the following structure theorem of Valiant~\cite{v82}, and its consequences (\autoref{clm:composed poly}). Below, we state the theorem, and then use it to prove~\autoref{thm:vnp closure}. For completeness, we include a proof using the depth reduction results in~\cite{vsbr83} in the appendix. 

\begin{theorem}[Valiant~\cite{v82}]\label{thm:vnp structure}
Let $P(\vx)$ be a homogeneous polynomial of degree $r$ in $n$ variables that can be computed by an arithmetic circuit $C$ of size $s$. Then, there is an $m \leq \poly(s, r)$ and a polynomial $Q(\vx, y_1, y_2, \ldots, y_m)$ such that 
\[
P(\vx) = \sum_{\vy \in \{0,1\}^m} Q(\vx, \vy) \, , \text{ and }
\]
$Q(\vx, \vy)$ can be computed by an arithmetic formula of size $\poly(s, r)$.
\end{theorem}
We now proceed with the proof of~\autoref{thm:vnp closure}.
\begin{proof}[Proof of~\autoref{thm:vnp closure}]
Without loss of generality, we will assume that $P$ is monic in a variable $z$. This can be guaranteed by doing a linear transformation by replacing every variable $x_i$ by $x_i + a_iz$, where $a_i$ are chosen from a large enough grid, based on the degree of $P$. Note that this preserves the form of $P$ in the hypothesis of the theorem. Moreover, using~\autoref{thm:homog strassen}, we will assume that the degree of $Q(\vx, \vy)$ in the variables $x$ and $z$ is at most $r$, up to a polynomial blow up in the circuit size of $Q$. 

From~\autoref{lem:key lem gen}, we know that there is a $\vc \in \F^n$ and a polynomial $B$ in at most $t = O(d^2)$ variables, and polynomials $g_1, g_2, \ldots, g_t$ such that 
\[
f(\vx+ \vc, z) = \h_{\leq d}\left[B(g_1, g_2, \ldots, g_t) \right] \, .
\]
For the rest of this proof, we assume that we have shifted the origin, so that $\vc = {\bf 0}$. Again, this just requires replacing every variable $x_i$ by $x_i + c_i$, and this shift of coordinates does not affect the structure of $P$ in the hypothesis of the theorem. Thus, 
\[
f(\vx, z) = \h_{\leq d}\left[B(g_1, g_2, \ldots, g_t) \right] \, .
\]
Moreover, $B$ has a circuit of size at most $\poly(d)$ and each $g_i$ belongs to some set ${\cal G}_z(P, \alpha, d)$ for some $\alpha \in \F$. We now need the following two structural claims which follow from direct applications of properties of polynomials in $\VNP$ as shown by Valiant~\cite{v82}. 
\begin{claim}[Valiant~\cite{v82}]\label{clm:vnp generators}
For every choice of $\alpha \in \F$ and $g_j \in {\cal G}_z(P, \alpha, k)$, there is a polynomial $Q_j'(\vx, y_1, y_2, \ldots, y_m)$ such that 
\[
g_j(\vx) = \sum_{\vy \in \{0,1\}^m} Q_j'(\vx, \vy) \, .
\]
Moreover, $Q'$ can be computed by a circuit of size at most $\poly(s, r, d)$. 
\end{claim}
The second claim is about the structure of the composed polynomial $B(g_1, g_2, \ldots, g_t)$. This is a special case of a more general result of Valiant~\cite{v82}, which showed that $\VNP$ is closed under \emph{composition}.
\begin{claim}[Valiant~\cite{v82}]\label{clm:composed poly}
There is an $\tilde{m} \leq \poly(m, d)$ and a polynomial $\tilde{Q}(\vx, y_1, y_2, \ldots, y_{\tilde{m}})$ such that 
\[
B(g_1, g_2, \ldots, g_t) = \sum_{\vy \in \{0,1\}^{\tilde{m}}} \tilde{Q}(\vx, \vy) \, .
\]
Moreover, $\tilde{Q}$ can be computed by a circuit of size $\poly(s, r, d, n, m)$. 
\end{claim}
For completeness, we provide a sketch of the proofs of the claims and that of~\autoref{thm:vnp structure} to the appendix. We now use the claims above to complete the proof of~\autoref{thm:vnp closure}.

Observe that if we view $\tilde{Q}$ as a polynomial in $\vx$ variables with coefficients coming from $\F[\vy]$, then, for every $k \in \N$, it follows that 
\[
\h_{k}\left[B(g_1, g_2, \ldots, g_t)\right] = \sum_{\vy \in \{0,1\}^{\tilde{m}}} \h_{k, \vx}\left[\tilde{Q}(\vx, \vy) \right]\, .
\]
Here, $\h_{k, \vx}[\tilde{Q}(\vx, \vy)]$ denotes the homogeneous component of degree $k$ of $\tilde{Q}(\vx, \vy)$ when viewing $\tilde{Q}(\vx, \vy)$ as a polynomial in $x$ variables. It follows from~\autoref{thm:homog strassen}, that by blowing up the size of the circuit for $\tilde{Q}$ by a factor of at most $O(k^2)$, we can obtain a circuit which computes $\h_{k, \vx}[\tilde{Q}(\vx, \vy)]$, and this does not affect the $y$ variables in any way. This gives us a representation of $f(\vx, z)$ as 
\[
f(\vx) = \sum_{\vz \in \{0,1\}^{m'}} h(\vx, \vz) \, .
\]
where $m' = \tilde{m} \leq \poly(m, d)$, and $h$ can be computed by a circuit of size $\poly(s, r, d, n, m)$. This completes the proof of the theorem. 
%To get a representation for $f(\vx, y)$, we just replace every variable $x_i$ in the circuit for $\tilde{Q}$, by $x_i - c_i$. This completes the proof of~\autoref{thm:vnp closure}.
\end{proof}
%The following theorem immediately follows from~\autoref{thm:vnp closure}. 

\section{Factors of Polynomials with Small Formulas}
%\Mnote{Need to include definitions of formulas and class VF.}
In this section, we prove the following theorem, which gives an upper bound on the formula complexity of factors of polynomials which have small formulas.  We note that this result is not new and was also proved by Dutta et al. in~\cite{DSS17}. %{\color{red}\bf CN: (Here I summarize the difference between their proof and our proof. Do you think it's necessary?)}
Since the proof essentially follows from our techniques developed so far and our proof is different from the proof in~\cite{DSS17}, we include the statement and a proof sketch. 
\begin{theorem}[\cite{DSS17}]
Let $P(\vx)$ be a polynomial of degree $r$ in $n$ variables which can be computed by an arithmetic formula of size at most $s$, and let $f(\vx)$ be a factor of $P$ of degree at most $d$. Then, $f(\vx)$ can be computed by an arithmetic formula of size at most $\poly(s, r, n, d^{O(\log d)})$. 
\end{theorem}
\begin{proof}
The proof is again along the lines of the proof of~\autoref{thm:vnp closure}. We first observe that the polynomials in ${\cal G}_y(P, \alpha, k)$ have small formulas. This just follows from the proof of Item 3 in~\autoref{lem:G properties} and~\autoref{lem:interpolation formulas}. 

Now, recall that from~\autoref{lem:key lem gen}, we know that the  $B$ is a polynomial in at most $O(d^2)$ variables, and can be computed by a circuit of size at most $\poly(d)$. Thus, by~\autoref{thm:circuit to formulas}, we get that $B$ can be computed by a formula $\Phi$ of size at most $d^{O(\log d)}$. Composing $\Phi$ with the formulas for the polynomials in ${\cal G}_y(P, \alpha, k)$, we get a formula for $B(g_1, g_2, \ldots, g_t)$ of size at most $\poly(r, s, m, n, d^{O(\log d)})$, and also, 
\[
f = \h_{\leq d}\left[B(g_1, g_2, \ldots, g_t) \right] \, .
\]

All we need now to complete the proof, is a formula for $\h_{\leq d}\left[B(g_1, g_2, \ldots, g_t) \right]$, and this follows from~\autoref{lem:interpolation formulas}.
\end{proof}
We remark that the proof also extends to the case of Algebraic Branching Programs.

\section*{Acknowledgment}
We thank Rafael Oliveira for making us aware of the question about the complexity of factors of polynomials in $\VNP$, and Guy Moshkovitz for helpful discussions.

\bibliographystyle{customurlbst/alphaurlpp}
\bibliography{references}

\appendix
\section{Proofs of claims}
%\section{Proofs of~\autoref{thm:vnp structure}, {~\autoref{clm:vnp generators}} and~\autoref{clm:composed poly}.}
%\paragraph{ }
We now include the proofs of~\autoref{thm:vnp structure}, \autoref{clm:vnp generators} and~\autoref{clm:composed poly}. We follow the notation set up in the proof of~\autoref{thm:vnp closure}.

\begin{proof}[Proof of~\autoref{clm:vnp generators}] We relabel one of the variables in $\vx$ as $z$.
Let $C_0(\vx), C_1(\vx), \ldots, C_{r}(\vx)$ be polynomials such that \[
P(\vx, z) = \sum_{i = 0}^{r} C_i(\vx) \cdot z^i \, .
\]
From~\autoref{def:derivative}, we know that $\frac{\partial^{j} P }{\partial z^j}\left(\vx, z \right)$ equals $\sum_{i = j}^r \binom{i}{j}C_i(\vx)\cdot z^{i-j}$. Now, we know that
\[
P(\vx, z) = \sum_{y \in \{0,1\}^m} Q(\vx, \vy, z) \, .
\]
Expressing $Q(\vx, \vy, z)$ as a univariate in $z$, we get 
\[
Q(\vx, \vy, z) = \sum_{i = 1}^{r} C'_i(\vx , \vy) \cdot z^{i} \, .
\]
Recall that $Q(\vx, \vy, z)$ has a circuit of size at most $\poly(s)$ and degree at most $r$. By viewing $Q$ as a univariate in $z$ and applying~\autoref{thm:homog strassen}, we get that each $C'_i(\vx, \vy)$ has a circuit of size $\poly(s, r)$. In particular, for every $j \in \N$, we can write $C_j(\vx)$ as 
\[
C_j(\vx) = \sum_{\vy \in \{0,1\}^m} C'_j(\vx, \vy) \, .
\] 
Therefore, for every $j \in \{0, 1, 2, \ldots, d\}$, we get 
\[
\sum_{i = j}^r \binom{i}{j}C_i(\vx) \cdot z^{i-j} =  \sum_{\vy \in \{0,1\}^m}\left(\sum_{i = j}^r \binom{i}{j}C'_i(\vx, \vy) \cdot z^{i-j} \right)
\]
Moreover, the polynomial $\left(\sum_{i = j}^r \binom{i}{j}C'_i(\vx, \vy) \cdot z^{i-j} \right)$ has a circuit of size at most $\poly(n, r)$. This completes the proof of the claim. 
\end{proof}

\begin{proof}[Proof of~\autoref{clm:composed poly}]
The proof is in two parts. We first define the construction of the circuit for $\tilde{Q}$, and then argue the correctness of this construction.
\paragraph*{Constructing $\tilde{Q}$. } 
We know that $B(z_1, z_2, \ldots, z_t)$ is of degree at most $d$ and can be computed by a circuit of size $\poly(d)$. It follows from~\autoref{thm:vnp structure}, that there is an $a \leq \poly(t,d)$ and a  polynomial $B'$ in at most $t + a$ variables such that 
\[
B(z_1, z_2, \ldots, z_t) = \sum_{\vy \in \{0,1\}^a} B'(\vz, \vy) \, .
\]
Crucially, it is also the case that $B'$ has a \emph{formula} of size at most $\poly(d, t)$. We remark that it is extremely important for the proof that $B'$ has a small formula, and not just a small circuit. To construct $\tilde{Q}$, we consider the formula $\Phi$ for $B'(\vz, \vy)$ and let $\ell_1, \ell_2, \ldots,\ell_u$ be the leaves of $\Phi$. Each of the leaves is labeled by a $z$ variables, a $y$ variable or a field constant. From this, we construct a circuit $\Phi'$ by going through over the leaves, and replacing the leaf $\ell_i$ by the circuit for polynomial $Q'_j(\vx, \vy_i)$ from~\autoref{clm:vnp generators} if it is labeled by $z_j$ and leaving it unchanged otherwise. Thus, $\Phi'$ computes polynomial in variables $\vx \cup \vy \cup_{j = 1}^u \vy_{j}$, of size at most $\poly(s,d,r, n, m)$. We denote this polynomial by $\tilde{Q}$. Let $\tilde{m} = \abs{\vx \cup \vy \cup_{j = 1}^u \vy_{j}} \leq \poly(m, d)$. We now argue that the construction in correct. 

\paragraph*{Correctness. } We now argue that 
\[
B(g_1, g_2, \ldots, g_t) = \sum_{(\vy, \vy_1, \ldots, \vy_u) \in \{0,1\}^{\tilde{m}}} \Phi'(\vx,\vy, \vy_1, \ldots, \vy_u) \, .
\]
The proof is by an induction on the size of formula $\Phi$ and the fact that in going from $\Phi$ to $\Phi'$, each of the leaves of $\Phi$ which was labeled by a $z_j$ variable was replaced by a copy of $Q'_j$ with a unique copy of the  auxiliary $y$ variables. Note that the uniqueness of the auxiliary variables is due to the fact that $B'$ has a formula. Finally, the proof follows from the following observation showing that $\tilde{m}=\poly(s,t,d)$. We skip the details.
\begin{observation}
Let $R_1(\vx), R_2(\vx)$ and  $S_1(\vx, \vy), S_2(\vx, \vz)$ be polynomials such that 
\[
R_1(\vx) = \sum_{\vy \in \{0,1\}^{\abs{\vy}}} S_1(\vx, \vy) \, , \text{and }
\]
\[
R_2(\vx) = \sum_{\vz \in \{0,1\}^{\abs{\vz}}} S_2(\vx, \vz) \, . 
\]
Then, 
\[
R_1(\vx) + R_2(\vx) = \sum_{\vy \in \{0,1\}^{\abs{\vy}}, \vz \in \{0,1\}^{\abs{\vz}}} \left(S_1(\vx, \vy) + S_2(\vx, \vz)\right)\, , \text{and }
\]
\[
R_1(\vx) \times R_2(\vx) = \sum_{\vy \in \{0,1\}^{\abs{\vy}}, \vz \in \{0,1\}^{\abs{\vz}}} \left(S_1(\vx, \vy) \times S_2(\vx, \vz) \right)\, .
\]
\end{observation}
%That is, the number of auxiliary variables grows polynomially after polynomially many compositions. 
\end{proof}
\begin{proof}[Proof of~\autoref{thm:vnp structure}]Let $C'$ be the circuit obtained by applying~\autoref{thm:vsbr original} to the circuit $C$. The idea is to inductively turn $C'$ into a formula while reducing the depth by half in every step. of the From the properties of $C'$, we get that 
\[
P = \sum_{i = 1}^{s'} A_{i,1}\cdot A_{i,2} \cdots A_{i,5} \, .  
\] 
Here $s'\leq \poly(s, n, d)$ is the size of $C'$, and every $A_{i,j}$ is a polynomial computed by a sub-circuit in $C'$ and the degree of $A_{i, j}$ is at most $d/2+1$. We introduce variables $\set{y_{i, j}, i \in [s'], j \in [5]}$. Let $R(\vy)$ be the following polynomial. 
\[
R(\vy) = \sum_{i = 1}^{s'}\left(y_{i,1}\cdot y_{i,2}\cdots y_{i,5}\right)\cdot \prod_{i' \neq i}\left( (1-y_{i', 1})(1-y_{i', 2})\cdots (1-y_{i', 5})\right)
\]
Observe that for $\vb \in \{0,1\}^{\abs{\vy}}$, $R(\vb)$ is $1$ if and only if there is an $i \in [s']$ such that $(b_{i,1}, b_{i,2}, b_{i,3}, b_{i,4}, b_{i,5})$ equals $(1, 1, 1, 1, 1)$ and for all $i' \in [s']$ with $i \neq i'$, $(b_{i',1}, b_{i',2}, b_{i',3}, b_{i',4}, b_{i',5}) = (0, 0, 0 , 0, 0)$, and zero otherwise. Moreover, $R(\vy)$ can be computed by an arithmetic formula of size at most $s'^2 = \poly(s)$. Now, observe that we can write the polynomial $P$ as follows. 
\[
P(\vx) = \sum_{\vy \in \{0,1\}^{5s'}} R(\vy) \cdot \prod_{j = 1}^5 \left(\sum_{i = 1}^{s'} A_{i, j}y_{i, j}\right) \, .
\]
Also, for every $j$, the polynomial $\sum_{i = 1}^{s'} A_{i, j}y_{i, j}$ is of degree at most $d/2 + 1$ and can be computed by a circuit of size at most $3s'$. This is true since each $A_{i,j}$ is computed by a sub-circuit of $C'$. Thus, we have expressed a degree $d$ polynomial, computable by a circuit of size $s'$ in terms of polynomials of degree at most $d/2+1$, and circuit complexity $3s'$. We have also had to incur an additional additive cost of $O(s'^2)$ for the formula computing $R$. The idea of the proof is to keep applying this reduction for $\log d$ iterations, such that the degree of each of the polynomials is at most a constant. Then, we compute these \emph{generating} polynomials by a formula by brute force.

We now argue that the number of $y$ variables introduced in the process, and the total size of the formula for the final verifier is still polynomially bounded in $s, d$. The number of auxiliary $y$ variables introduced is given by the following recurrence. 
\[
m(d, s') \leq 5s' + 5m(d/2 + 1, 3s') \, . 
\]
The size of the formula ${F(d, s)}$ is upper bounded by the following recurrence. 
\[
F(d, s') \leq c\cdot s'^2 + 5F(d/2 + 1, 3s'),\ , 
\]
where $c>0$ is some constant. It is not hard to see that both $m(d,s')$ and $F(d,s')$ are upper bounded by a fixed polynomial function of $d, s'$.
\end{proof}

\end{document}